%% file: 0-main.tex
\documentclass[11pt]{article}
\usepackage[letterpaper,margin=1in]{geometry}

\input{packages}

\input{header}

\bibliography{./main}

\title{The Knapsack Secretary Problem is Strictly Harder  \\ Than the Secretary Problem}

\author{Anonymous Author(s)}

\author{%
    \begin{tabular}{c c}
        \begin{tabular}{c}
            Eric Balkanski\\
            Columbia University\\
            \texttt{eb3224@columbia.edu}
        \end{tabular}
        &
        \begin{tabular}{c}
            Jason Chatzitheodorou\\
            Columbia University\\
            \texttt{ic2621@columbia.edu}
        \end{tabular}
        \\ \\
        \begin{tabular}{c}
            Dimitris Fotakis\\
            NTUA and Archimedes, Athena RC\\
            \texttt{fotakis@cs.ntua.gr}
        \end{tabular}   
        &
        \begin{tabular}{c}
            Thanos Tolias\\
            NTUA and Archimedes, Athena RC\\
            \texttt{thanostolias@mail.ntua.gr}
        \end{tabular}     
    \end{tabular}
}

\date{}

\begin{document}
\maketitle

\begin{abstract}
The knapsack secretary problem is a generalization of  the classical secretary problem where the accepted items must satisfy a knapsack constraint. A line of work has developed constant-competitive algorithms for this problem, with successive improvements culminating in the current best-known competitive ratio of $0.153$. A natural open question was whether the optimal $1/e$ competitive ratio for the classical secretary problem is also achievable for the knapsack secretary problem.

We answer this question negatively by showing that no $(1/e - 0.0001)$-competitive algorithm exists for the knapsack secretary problem. The analysis of the family of hard instances we construct proceeds in three steps. First, we reduce the cardinal problem on these instances to an almost-ordinal problem. Second, we formulate a linear program that captures the performance of almost-ordinal algorithms on this instance family. Finally, we exhibit a feasible dual solution whose objective value is strictly below $1/e$. We also give an algorithm that improves the best-known competitive ratio from $0.153$ to $0.178$.

\end{abstract}

\thispagestyle{empty} 
\clearpage
\setcounter{page}{1}

\input{./1-Introduction}

\input{./2-Preliminaries}

\input{./3-UpperBound}

\input{./4-Algorithm}

\section*{Acknowledgements}

We gratefully acknowledge helpful discussions with Kevin Schewior.

\newpage

\printbibliography

\newpage

\appendix
\crefalias{section}{appendix}
\crefalias{subsection}{appendix}
\crefalias{subsubsection}{appendix}
\crefname{appendix}{Appendix}{Appendices}
\Crefname{appendix}{Appendix}{Appendices}

\input{./A-AbelsImpossibility}

\input{./B-MissingProofsUpperBound}
\input{./C-MissingProofsAlgorithm}

\end{document}

%% file: packages.tex
\usepackage{graphicx} 

\usepackage{amsmath, amsfonts, amssymb, amsthm, bbm}
\usepackage{easyReview}
\usepackage[
  backend=biber,
  natbib=true,
  maxbibnames=99,
  maxcitenames=99
]{biblatex}
\usepackage[autostyle, english = american]{csquotes}
\MakeOuterQuote{"}
\usepackage{mathpazo}
\usepackage{booktabs}
\usepackage[linesnumbered,ruled,vlined,noend]{algorithm2e}
\usepackage[dvipsnames]{xcolor}

\definecolor{myRed}{rgb}{0.82,0.13,0.56}
\definecolor{myBlue}{RGB}{13,55,174}
\usepackage[colorlinks=true,citecolor=myRed,linkcolor=myBlue,urlcolor=PineGreen]{hyperref}

\usepackage{thmtools}
\usepackage{thm-restate}
\usepackage[noabbrev,capitalise]{cleveref}

\usepackage{todonotes}
\usepackage{comment}

%% file: header.tex
\declaretheorem[name=Theorem]{theorem}
\declaretheorem[name=Lemma, sibling=theorem]{lemma}

\declaretheorem[name=Corollary, sibling=theorem]{corollary}
\declaretheorem[name=Remark, sibling=theorem]{remark}
\declaretheorem[name=Definition, sibling=theorem]{definition}

\SetAlFnt{\small}
\SetAlCapFnt{\small}
\SetAlCapNameFnt{\small}
\SetAlCapHSkip{0pt}
\IncMargin{-\parindent}

\newcommand{\opt}{\textnormal{\textsc{Opt}}}
\newcommand{\alg}{\textnormal{\textsc{Alg}}}
\newcommand{\ordalg}{\textnormal{\textsc{Alg}}_o}
\newcommand{\cardalg}{\textnormal{\textsc{Alg}}_c}
\newcommand{\cardfam}{\mathcal{A}_c}
\newcommand{\ordfam}{\mathcal{A}_o}
\newcommand{\simu}{\textnormal{\textsc{Sim}}}
\newcommand{\first}{\textnormal{\textsc{First}}}

\newcommand{\lp}{\left} 
\newcommand{\rp}{\right}

\newcommand{\ebcomment}[1]{\textcolor{blue}{Eric: #1}}

\newcommand{\cD}{\mathcal{D}}

\newcommand{\cF}{\mathcal{F}}

\newcommand{\cI}{\mathcal{I}}

\newcommand{\cR}{\mathcal{R}}

\newcommand{\cT}{\mathcal{T}}

\newcommand{\largebinom}[2]{\left( \genfrac{}{}{0pt}{}{#1}{#2} \right)}

\newcommand{\E}[2]{\underset{#1\!}{\mathbb{E}}\lp[#2\rp]}

\newcommand{\Prob}[1]{\Pr\lp[#1\rp]} 

%% file: 1-Introduction.tex
\section{Introduction}

The secretary problem is a fundamental problem in computer science with a long and rich history. Its central role in online and approximation algorithms is reflected in the many lines of work studying the problem and its variants, including extensions to matroid constraints~\cite{babaioff2007matroids}, knapsack constraints~\cite{Babaioff07}, submodular value functions~\cite{bateni2013submodular}, stable matchings~\cite{babichenko2019stable}, prophet inequalities~\cite{azar2014prophet,esfandiari2017prophet}, non-uniform arrivals~\cite{kesselheim2015secretary}, pricing~\cite{cohen2014pricing}, online linear programming~\cite{agrawal2014dynamic}, and online ad allocation~\cite{feldman2010online}. While optimal competitive ratios have been established for many variants, determining the optimal competitive ratio remains a central open problem for several others. One prominent example is the strong matroid secretary conjecture, a landmark open problem in online algorithms asserting that there exists a $1/e$-competitive algorithm for matroid secretary, matching the optimal guarantee for the classical secretary problem.

Another important family of feasibility constraints for  which it is unknown if a $1/e$-competitive algorithm exists is knapsack constraints. In the knapsack secretary problem, there are  $n$ items, each with a value and a size, that arrive in a random order, as well as a knapsack capacity. Upon arrival of an item, its value and size are revealed and the algorithm must make the irrevocable decision of either accepting this item or rejecting it. The goal is to maximize the total value of accepted items, under the constraint that their total size is at most the knapsack capacity.

The problem was introduced by~\citet{Babaioff07}, who gave the first constant-competitive algorithm, with competitive ratio $1/(10e)$. This guarantee was subsequently improved to $1/8.06$ by~\citet{Kesselheim14}, who obtained this ratio for the online generalized assignment problem (GAP) in the random-order model using their framework for online packing LPs; to $1/6.65$ by~\citet{albers2020}; and to the current best guarantee of $(1-\ln 2)/2 \approx 1/6.52$ by~\citet{klimm2025}, whose online GAP algorithm also applies to knapsack secretary as a special case. Moreover, \citet{klimm2025} proved a $1/e$-competitive ratio for the fractional knapsack secretary problem and that such a competitive ratio is the best possible for deterministic algorithms.

 Despite this progress, the best-known impossibility for knapsack secretary, and more generally online GAP in the random order model, remains the $1/e$ impossibility inherited from the classical secretary problem. For the $1$--$2$-knapsack problem, which is the special case where items have sizes $1$ or $2$ and the capacity is $2$, a $1/e$-competitive ratio has been achieved by~\citet{abels2022knapsack}, who noted that, for the general knapsack secretary problem, no impossibility beyond $1/e$ had been found and that their result left open the possibility that the same ratio might also be attainable.

\begin{center}
   \emph{Is there a $1/e$-competitive algorithm for the knapsack secretary problem?}
\end{center}

Our main result is a negative answer to this question.

\begin{restatable}{theorem}{TheoremOne}
\label{thm: Theorem 1}
    There is no $(1/e - 0.0001)$-competitive algorithm for the knapsack secretary problem. 
\end{restatable}

This information-theoretic impossibility holds for both deterministic and randomized algorithms. It is also, to the best of our knowledge, the first to rule out a $1/e$-competitive algorithm for online GAP in the random order model.  We  note that \citet{abels2022knapsack} show that value-ordinal algorithms, which  are algorithms that do not observe the cardinal values of the items but instead
only observe their relative order, cannot achieve a competitive ratio better than $1/(e+1) + o(1)$. We provide a further discussion of this result and how its analysis compares to ours in Appendix~\ref{app:previouswork}. The analysis of our family of hard instances proceeds in three steps. First, we reduce the cardinal problem on these instances to an ordinal problem where the algorithm only observes the relative order of high-value items instead of their cardinal values. Second, we formulate a linear program that captures the performance of ordinal algorithms on our   family of hard instances. Finally, we exhibit a dual solution with objective value strictly below $1/e$, and verify its feasibility computationally.

 In addition, we also improve the best-known competitive ratio for the problem.

\begin{theorem}
    There is a polynomial-time $1/5.59$-competitive algorithm for knapsack secretary. 
\end{theorem}

This $1/5.59$ competitive ratio improves over the previous best competitive ratio of $1/6.52$ \cite{klimm2025}. Our positive result is obtained by a three-phase algorithm that first targets high-value items through a secretary-like value-record rule and, if no such item is selected, switches to a dynamic density-greedy rule for small items. Our dynamic density-greedy algorithm bears a similarity to the virtual algorithm of~\citet{Babaioff07} and the fractional algorithm of~\citet{klimm2025}, while our analysis is fundamentally different.  For the analysis, we  combine itemwise selectability with a carefully formulated charging argument. We expand on this approach in the technical overview below.

\paragraph{Concurrent work.} Independently, \citet{garbea2026} also show that no $1/e$-competitive algorithm exists for the knapsack secretary problem. More precisely, they prove that no $(1/e - 0.0035)$-competitive algorithm exists.

\subsection{Technical Overview} 

\noindent\textbf{The Impossibility Result.} 
We define a class of instances with a unit size knapsack and two kinds of items. The \emph{small items} have size $1/n$ and value $1$ and the \emph{large items} have size \(1\) and much larger values. The values of the large items are chosen to satisfy two properties: (i)~the value of every large item is  much larger than $n$ (the maximum total value collected by small items), so selecting the best large item is overwhelmingly more important than collecting small items; and (ii)~the values of different large items are separated by (extremely large) exponential gaps. 

For this family of instances, the competitive ratio of any algorithm takes a simple form and is  equal to the minimum of the following (up to negligible lower-order terms): for the all-small instance, the fraction of small items accepted by the algorithm; for any other instance, the probability that the highest-value large item is selected. Thus, our instances force an algorithm to balance two opposing requirements: it should start accepting small items early enough to be competitive on the all-small instance, but doing so may prevent it from accepting a large item that appears later.

A key observation is that for the family of instances above, the optimal algorithm should behave essentially as an ordinal algorithm. The first part of the proof, in Section~\ref{sec: reduction}, formalizes this intuition by upper bounding the competitive ratio for the \emph{cardinal objective} of maximizing the expected total value of items accepted by the algorithm by the \emph{ordinal benchmark} of maximizing the probability that the highest-value large item is accepted (with a small tweak concerning the all-small instance, plus lower-order terms). 

Bounding the original cardinal objective from above by the ordinal benchmark is based on the extreme separation of the large item values. By the value separation properties (i) and (ii) above, the precise magnitudes of the values of the large items matter only through their rank: in instances with large items, we have to select the best large item. So, we may replace the cardinal objective by an ordinal one, losing only negligible error terms. The more delicate step is to argue that cardinal algorithms do not gain any advantage over ordinal algorithms for this ordinal benchmark. A cardinal algorithm observes the actual values of previously revealed items, and these values could in principle leak information about the number of large items in the instance or about the rank of future large items. 
We build on the framework of~\citet{gravin2023onlineordinalproblemsoptimality}, using an adaptation of their order-statistics-indistinguishable (OSI) distributions. Specifically, we show (Lemma~\ref{lem: instances existence of osi distribution}) that there is a set of values satisfying our value separation properties, together with a distribution over these values, that is order-statistics-indistinguishable. Using this property, we show (Lemma~\ref{lem: instances reduction to ordinal algorithms}) that for every cardinal algorithm, there is an ordinal algorithm (Algorithm~\ref{alg: almost ordinal simulaiton}), with sample access to the distribution and black-box access to the cardinal algorithm, whose performance against the ordinal benchmark differs by a negligible term from the competitive ratio of the cardinal algorithm.

Having this result, we can focus on ordinal algorithms evaluated by the ordinal benchmark for the particular family of instances. The state of such an algorithm is described by the total number of items observed so far, the number of large items observed so far, and whether the current large item is best among the large items seen so far. 
Intuitively, a deterministic ordinal algorithm for this class of instances should behave roughly as follows. 

\begin{enumerate}
    \item If no large item is encountered among the first $q n$ items, the algorithm accepts the remaining small items for a competitive ratio of $1-q$, if the instance is all-small, and $0$ otherwise. Assuming that the instance contains $h \geq 1$ large items, the probability of the latter is $q^h$.  

    \item If a large item is encountered no later than the first $q n$ items are considered, the algorithm should switch to an ordinal algorithm for a secretary problem with an unknown number $h \geq 1$ of large items and uniform arrival times. 
\end{enumerate}

For any fixed value of $h \geq 1$, the strategy above results in a competitive ratio of at most $\min\{ 1-q, (1-q^h)\cdot \mathrm{ord}(h) \}$, where $\mathrm{ord}(h)$ is the best possible competitive ratio of the secretary problem with an unknown number $h$ of items (see \cite[Table~I]{stewart81}). Solving this optimally, we get an upper bound that tends to $1/e$ from above as $h$ increases and does not give any improvement.

However, this high-level intuition does not account for the extra pressure put to such an algorithm towards accepting the present large item, if it is best and appears relatively late in the sequence. For a precise analysis, we formulate a linear program (LP) that captures the behaviour of all ordinal algorithms on our family of instances. This follows an LP-based approach that has previously appeared in the study of secretary problems: the decisions available to an algorithm are encoded as LP variables, and the dual is used to certify upper bounds on the competitive ratio of any algorithm; see, e.g., \citet{BuchbinderLP14} and subsequent applications of factor-revealing LPs to secretary-type settings~\citep{chan2014revealing,correa2024sample,DuttingLP21,BalkanskiLP24,Abels_2025}. 

We use two types of LP variables to encode the decisions of any ordinal algorithm in our setting (see also Table~\ref{tab:primal_dual}). The variables $\alpha_{t,j}$ describe the probability of accepting a best-so-far large item after observing $t$ items in total and $j+1$ large items among them, conditional on the algorithm not having committed earlier. The variables $\beta_t$ describe the probability of switching to the small-item strategy after observing $t$ items and no large items among them. The constraints enforce feasibility of the induced strategy and require the same competitive ratio against every instance in the family. 

There exists an optimal ordinal algorithm that induces a feasible solution to this LP, so the LP optimum is an upper bound on the best competitive ratio achievable by any ordinal algorithm. 
We analyze this LP through its dual. For a fixed finite value of $n$, we construct an explicit feasible dual solution whose objective value is strictly below $1/e$, with enough slack to absorb the error terms introduced by the cardinal-to-ordinal difference. By weak duality, the primal optimum (and the best possible competitive ratio) is strictly below $1/e$, which holds also for cardinal algorithms.


\smallskip\noindent\textbf{Improving the Lower Bound.} 
Our positive result is motivated by the two different ways in which a knapsack optimum can obtain value. In some instances, the optimum is essentially determined by a single very valuable item and the problem resembles the classical secretary problem.  In other instances, the optimum is obtained by packing many smaller items, and density becomes the relevant criterion. We call
an item small if its size is at most \(1/2\). The analysis formalizes this
dichotomy through the following upper bound (Lemma~\ref{lem:decomp}):
\[
    \opt \le v(DG(1)) + v_z\,,
\]
where \(DG(1)\) is the final density-greedy solution on small items and \(z\) is the most valuable item not in \(DG(1)\). The algorithm obtains value from these two terms separately: an intermediate secretary-like phase targets the high-value item, while a final density-greedy phase targets the items in \(DG(1)\).
For the algorithm and its analysis we adopt the equivalent continuous-time formulation, where each item $i$ draws its arrival time \(t_i\) independently and uniformly from \([0,1]\). 

The algorithm (see also Algorithm~\ref{alg:positive}) has three phases, defined by the corresponding times $a, b \in [0,1]$, with $a < b$. The first phase is exploratory: the algorithm rejects all arriving items and uses them only to calibrate later decisions. The second phase uses the initial sample as a value threshold and accepts the first item with larger value, if such an item appears. If no item is accepted in this phase, the algorithm enters the final phase. In the final phase, when a small item \(i\) arrives at time \(t_i\), the algorithm recomputes the density-greedy solution \(DG(t_i)\) on the small items observed so far. If \(i\in DG(t_i)\), the algorithm attempts to accept it. Item $i \in DG(t_i)$ is accepted only if it fits in the remaining capacity. The algorithm's acceptance rule is tailored to the reference solution $DG(1)$: if \(i\in DG(1)\), then no matter when \(i\) arrives during the final phase, it also belongs to \(DG(t_i)\). Thus the algorithm attempts every reference item \(i\in DG(1)\) arriving in the final phase; the only obstruction is that previously accepted items may have consumed too much capacity.

The expected value collected in the second phase is analyzed by standard tools (Lemmas~\ref{lem:second-phase} and \ref{lem:phase-three-probability}). The improvement comes from a careful analysis of the expected value collected in the final phase. To illustrate the main ideas, we first provide (Section~\ref{sec:PhaseIII}) a simple uniform analysis of the probability that each $i \in DG(1)$ that arrives in the final phase is accepted (a.k.a. $i$'s selectability). 

Since every item $i \in DG(1)$ has size at most \(1/2\), a sufficient condition for accepting $i$ is that at least half of the knapsack capacity is available at $i$'s arrival time. Therefore, since the algorithm attempts to accept every item in \(DG(1)\) that arrives in the final phase, our selectability analysis reduces to upper bounding the expected total size of accepted items appearing before a fixed item \(i\in DG(1)\) arrives. 
For any time \(t\), the items having arrived by time $t$ form a random sample (Lemma~\ref{lem:phase-three-conditioning}). Using this and the fact that at any time $t$, the density-greedy solution is feasible, and thus its total size is at most $1$, we can show that for any fixed time $t$ and any item $i$, the probability that $i$ is included in the density-greedy solution $DG(t)$ is at most $1/t$. 
%
%
Hence, we obtain an upper bound of $1/t$ on the expected total size of accepted items by time \(t\). Integrating over the time interval from the beginning of the final phase $b$ up to the moment when item \(i\) arrives, we get a logarithmic (in $i$'s arrival time) upper bound on the total size of the items accepted before $i$ (Lemma~\ref{lem:expected-attempted-mass}). Combined with the expected value collected in the second phase, the resulting probability that such an item $i$ is included in the final solution gives an $1/7$-competitive algorithm.

Next, in Section~\ref{section:Breaking}, we refine the analysis above, using that the probability that an item fits in the current solution depends on its  size. An item of size \(s\) only needs \(1-s\) remaining capacity; hence a smaller item $i \in DG(1)$ that arrives in the final phase should be accepted with larger probability. After lower bounding the selection probability of such an item $i$ as a function of its size $s_i$ (Lemma~\ref{lem:size-sensitive-selectability}), we compare it with the target competitive ratio and view any shortfall as a \emph{deficit} (Definition~\ref{def:deficit}). Since smaller items have larger selectability, we conclude that only relatively large items can have a significant deficit. Because the density-greedy solution has total size at most $1$, the total deficit over all items of \(DG(1)\) is small (and upper bounded in Corollary~\ref{cor:worst-case-deficit}). For the items whose deficit is large, we show that the secretary-like second phase contributes the remaining amount of selectability: beyond selecting the maximum-value item with positive probability, the second phase also has positive probability of selecting other relatively highly ranked items. We charge these probabilities to the high-value item \(z\) and to the total deficit incurred by the density phase. This careful (and to the best of our knowledge, novel) coupling between size-dependent selectability in the final phase and rank-based selectability in the second phase results in the improved competitive ratio.

%% file: 2-Preliminaries.tex
\section{Preliminaries}



An instance of the knapsack secretary problem consists of $n$ items $[n] = \{1, \ldots, n\}$, each with a value $v_i \geq 0$ and a size $s_i \in (0, 1]$.  The values and  sizes of the items are initially  unknown to the algorithm. The items arrive online in a uniformly random order $\pi$, where $\pi(t)$ denotes the item in position $t$ and $\pi[t]$ the set of items in positions $[t]$. Upon arrival of an item $i$, the algorithm observes its size $s_i$ and  value $v_i$, and must immediately and irrevocably decide whether to accept or reject it. The knapsack constraint is that the total size $s(A) = \sum_{i \in A} s_i$ of the set $A$ of accepted items must satisfy $s(A)  \le 1.$ The objective is to maximize the total accepted value $v(A) = \sum_{i \in A} v_i.$

 An  algorithm $\alg$ is $c$-competitive if, for every instance $(s, v)$, $$\E{\pi, \alg}{v(\alg(s, v, \pi))} \ge c \cdot \operatorname{OPT}(s, v),$$
 where $\alg(s, v, \pi)$ is the set of items accepted by algorithm $\alg$ over instance $(s, v)$ and ordering $\pi$, and $\operatorname{OPT}(s, v)
    = v(O(s,v))$
   is the value of optimal solution $ O(s,v) = \arg\max_{S \subseteq [n] : s(S) \le 1} v(S)$.
For the remainder of the paper, when the algorithm is clear from context, we omit $\alg$ from the subscripts of expectations and probabilities.

%% file: 3-UpperBound.tex
\section{The Impossibility Result}

In this section, we show that no algorithm can achieve a $1/e - 0.0001$ competitive ratio for the knapsack secretary problem. First, in \cref{sec: instances}, we describe the family of hard instances used to establish this impossibility result. 
In \cref{sec: reduction}, we show that for these hard instances, observing the actual values of the items offers almost no advantage over ordinal algorithms that observe only their relative ranks, except for small items. 
In \cref{sec: properties}, we provide properties that must be satisfied by an optimal ordinal algorithm. We then use these properties to formulate the performance of optimal ordinal algorithms as a linear program and use weak duality in \cref{sec: linear program} to upper bound the performance of ordinal algorithms. Lastly, in \cref{sec: proof of theorem 1}, we complete the proof of \cref{thm: Theorem 1} by combining the main lemmas from \cref{sec: reduction}, \cref{sec: properties}, and \cref{sec: linear program}.

\subsection{The Family of Hard Instances}
\label{sec: instances}

The family of hard instances is:
\begin{align*}
\mathcal I = \{(s,v) : \ & n \geq 14, h \in \{0, 1, \ldots, 14\}, \\
& s_i = 1 \text{ and } v_i > 1 \text{ for } i \in \{1, \ldots, h\}, \\
& s_i = 1/n \text{ and } v_i = 1 \text{ for } i \in \{h+1, \ldots, n\}\} 
\end{align*}

There are two types of items: the large items of size $1$ and the small items of size $1/n$. Note that the subinstances of size $h$ that consist of only  large items are the family of instances with $h$ items of the classical (valued) secretary problem where a single item can be selected. The number of large items $h$ is unknown to the algorithm. In particular, there could be no large items, in which case the algorithm needs to start accepting small items at time $(1-1/e)n$ if it has not yet seen a large item and wants to achieve a $1/e$ competitive ratio on the instance with no large items. However, once the algorithm accepts a single small item, it cannot select a large item anymore.

Next, we define a mapping from a set of values $V$ to an instance $(s, v) \in \mathcal{I}$. Let $h$ be the number of values in $V$ strictly larger than $1$. For $i \in \{1, \ldots, h\}$, we let $s_i = 1$ and $1< v_1 < \cdots < v_h$ be these $h$ large values in $V$ strictly larger than $1$, in increasing order. For $i \in \{h+1, \ldots, n\}$, we let $s_i = 1/n$ and $v_i = 1$. For the remainder of this section, we thus abuse notation and write $\alg(V, \pi)$ and $O(V)$ for $\alg(s,v, \pi)$ and $O(s,v)$ where $(s,v)$ is the  corresponding instance obtained from $V$ with the above mapping.

\subsection{The Reduction to Almost-Ordinal Algorithms}
\label{sec: reduction}

We start by defining the family of almost-ordinal algorithms.
An almost-ordinal algorithm is an algorithm that does not observe the values $v_i > 1$ of large items $i \leq h$, but instead only observes the relative order of these values. These algorithms still get to observe the sizes of all the items and the values $v_i = 1$ of items $i > h$.
\begin{definition}
\label{def: almost ordinal algorithm}
    An algorithm $\ordalg$ is in the family of almost-ordinal algorithms $\ordfam$ if, for the items $i \in [h]$, upon arrival of $i$ at time $t$, $\ordalg$ does not observe $v_i$ but instead observes its value's relative rank $1 + |\{t' \in [t]: v_{\pi(t')} > v_i\}|$ among the item values that have arrived so far.
\end{definition}

Note that the instances in the family $\mathcal I$ that have the same number of large items $h$ are indistinguishable to an almost-ordinal algorithm. Thus, we let 
$\ordalg(h, \pi)$ denote the items selected by an almost-ordinal algorithm $\ordalg \in \ordfam$ over an instance in $\mathcal I$ with $h$ large items. Next, we define  the following auxiliary objective $\rho(\ordalg)$ for an algorithm $\ordalg \in \ordfam$: 
$$\rho(\ordalg) = \min\left( \min_{h \in \{1, \ldots, 14\}} \Pr_{\pi}\left[ h \in \ordalg(h, \pi) \right], \frac{1}{n} \cdot \E{\pi}{|\ordalg(0, \pi)|}  \right).$$

This auxiliary objective $\rho(\ordalg)$ is a lower bound on the competitive ratio achieved by $\ordalg$ over instances in $\mathcal I$ since
\begin{itemize}
    \item for instances in $\mathcal I$ with $h \geq 1$ large items, $\Pr_{\pi}\left[ h \in \ordalg(h, \pi) \right]$ is the probability that $\ordalg$ accepts the optimal solution $\{h\}$, which is the large item of largest value $v_h$, and 
    \item for instances in $\mathcal I$ with $h = 0$ large items, $|\ordalg(0, \pi)|$ is the number of small items accepted by $\ordalg$ and the optimal solution is to accept all $n$ small items.
\end{itemize}

We now formally state the reduction, which is the main lemma of this section. It shows that the optimal competitive ratio achievable by a cardinal algorithm $\cardalg \in \cardfam$, where we let $\cardfam$ denote all cardinal algorithms, over instances $\mathcal I$ is, up to lower order terms, at most the optimal auxiliary objective $\rho(\ordalg)$ achievable by an almost-ordinal algorithm $\ordalg \in \ordfam$.
\begin{lemma}
\label{lem: instances reduction to ordinal object ordinal algorithm}
    We have that:
    \begin{align*}
        \max_{\cardalg \in \cardfam} \min_{(s,v) \in \mathcal I} \frac{\E{\pi}{v\left( \cardalg(s, v, \pi) \right) }}{v( O(s,v))} \leq 
         \left(1 + \frac{1}{n^2}\right) \cdot \max_{\ordalg \in \ordfam} \rho(\ordalg) + \frac{2}{n^2}.
    \end{align*}
\end{lemma}
\vspace{30pt}
The remainder of Section~\ref{sec: reduction} is devoted  to proving \cref{lem: instances reduction to ordinal object ordinal algorithm}. There are two main steps:
\begin{enumerate}
    \item In \cref{lem: instances reduction to ordinal objective}, we show that on instances with $h \geq 1$ and exponential gaps between the large values, the competitive ratio of cardinal algorithms is, up to lower order terms, the probability that the highest value item is selected.
    \item In \cref{lem: instances reduction to ordinal algorithms}, we show that there is a distribution $\mathcal F$ over sets of $n$ values such that, for any cardinal algorithm $\cardalg$, one can construct an almost-ordinal algorithm $\ordalg$ that simulates $\cardalg$ and satisfies the following guarantee: the probability that $\ordalg$ selects the highest-value item on an instance drawn from $\mathcal F$ is, up to lower-order terms, at least the corresponding probability for $\cardalg$. 
    To prove this, we adapt techniques from \citet{gravin2023onlineordinalproblemsoptimality}, specifically their use of order-statistics-indistinguishable (OSI) distributions. There are two reasons why we cannot use their result in a black box fashion:  we require a distribution with sampled values that have exponential gaps, which is done in \Cref{lem: instances multiplicative gap osi distribution}, and almost-ordinal algorithms are algorithms that are not fully ordinal.
\end{enumerate}

We proceed to the first step. Let $W = \{w_1, \ldots, w_h\}$ with $w_1 < \cdots < w_{h}$, then we let  $W \cup \{1\}^{n-h}$ denote the multiset  $\{v_1, \ldots, v_n\}$ with $h$ large items with values  $v_i = w_i$ for $i \leq h$ and $n - h$ small items with values $v_i = 1$ for $i > h$.  The next lemma shows that, for instances $V \in \left\{ W \cup \{1\}^{n-h}: W \in \largebinom{\{n^{3k}: k \in \mathbbm{N}\}}{h}\right\}$ with exponential gaps between the large values, the competitive ratio of cardinal algorithms is, up to lower order terms, the probability that the highest-valued item $h$ is selected.
\begin{lemma}
\label{lem: instances reduction to ordinal objective}
    For all $h \in [14]$, $\cardalg \in \cardfam$, $V \in \left\{ W \cup \{1\}^{n-h}: W \in \largebinom{\{n^{3k}: k \in \mathbbm{N}\}}{h}\right\}$ we have:
    \begin{align*}
        \frac{\E{\pi}{v\left( \cardalg(V, \pi) \right)}}{v(O(V))}
        \leq \Pr_{\pi}\left[ h \in \cardalg(V, \pi) \right] + \frac{2}{n^2}.
    \end{align*}
\end{lemma}
\begin{proof} 
    First note that $O(V) = \{h\}$, since we can only accept at most $1$ of the large items due to their size, and accepting any number of small items can give total value at most $n$ while $v_h = \max{V}$ and $v_h > n^3 \cdot v_1 > n^3$ by construction, therefore $v(O(V)) = v_h$. Then we have:
    \begin{align*}
        &\E{\pi}{v\left( \cardalg(V, \pi) \right)} \\
        =\ &\Pr_{\pi}\left[ h \in \cardalg(V, \pi) \right] \cdot v_h + \sum_{i=1}^{h-1} \Pr_{\pi}\left[ i \in \cardalg(V, \pi) \right] \cdot v_i \\ &+ \Pr_{\pi}\left[ \{1, \dots, h\} \cap \cardalg(V, \pi) = \emptyset \right] \cdot \E{\pi}{ v(\cardalg(V,\pi)) \middle| \{1, \dots, h\} \cap \cardalg(V, \pi) = \emptyset } \\
        \leq\ &\Pr_{\pi}\left[ h \in \cardalg(V, \pi)  \right] \cdot v_h + h \cdot \frac{v_h}{n^3} + 1 \cdot n \\
        \leq\ &\Pr_{\pi}\left[ h \in \cardalg(V, \pi) \right] \cdot v_h + \frac{1}{n^2} \cdot v_h + \frac{1}{n^2} \cdot v_h \\
        =\ &\left( \Pr_{\pi}\left[ h \in \cardalg(V, \pi) \right] + \frac{2}{n^2} \right) \cdot v_h, \tag{1}
    \end{align*}
    where the first line is by law of total probability since any algorithm that respects the size constraint can either accept one large item or multiple small items, the second line is by upper bounding the probabilities with $1$, $v_i \leq v_h/n^3$ by construction and $v(\cardalg(V,\pi)) \leq n$ when $\{1, \dots, h\} \cap \cardalg(V, \pi) = \emptyset$, i.e. accepts small items, the third line is because $h \leq n$ and $v_h \geq n^3$ and the fourth line is by factoring.
    Dividing both sides of (1) by $v(O(V))$ we have:
    \begin{align*}
        \frac{\E{\pi}{v\left( \cardalg(V, \pi) \right)}}{v(O(V))} 
        \leq \left( \Pr_{\pi}\left[ h \in \cardalg(V, \pi) \right] + \frac{2}{n^2} \right) \cdot \frac{v_h}{v(O(V))} 
        = \Pr_{\pi}\left[ h \in \cardalg(V, \pi) \right] + \frac{2}{n^2},
    \end{align*}
    where the equality is since $v(O(V)) = v_h.$
\end{proof}

In order to define  order-statistics-indistinguishable  distributions, we give the definition of total variation distance and show two standard properties, whose proofs we defer to \cref{sec: appendix B} for completeness.
\begin{definition}
    Let $X, Y$ be random objects sampled from probability mass functions $p_X, p_Y$ over discrete domain $\cT$. Then the total variation distance between $X, Y$ is:
    \begin{align*}
        d_{TV}(X, Y) := \frac{1}{2} \sum_{t \in \cT}|p_X(t) - p_Y(t)|.
    \end{align*}
\end{definition}

\begin{restatable}{lemma}{TVTriangleInequality}
\label{lem: instances total variation triangle}
    Let $X, Y, Z$ be random objects over discrete domain $\cT$. Then $d_{TV}(X, Y) \leq d_{TV}(X, Z) + d_{TV}(Z, Y)$.
\end{restatable} 
\begin{restatable}{lemma}{TVMapping}
\label{lem: instances total variation mapping}
    Let $X, Y$ be random objects over discrete domain $\cT$ and $f : \cT \rightarrow \cR$ be an arbitrary, potentially randomized, mapping  whose internal randomness is independent of $X, Y$.
    Then $d_{TV}(f(X), f(Y)) \leq d_{TV}(X, Y)$.
\end{restatable}

We now define the class of order-statistics-indistinguishable (OSI) distributions. 
The definition is identical to that of \citet{gravin2023onlineordinalproblemsoptimality}, but with a $1/n^5$ bound that suffices for our proof.
Intuitively, in OSI distributions, observing some set of order statistics does not reveal any information about which order statistics they are.
In this way, OSI distributions effectively ``hide'' the future values from an algorithm that observes them in random order.
\begin{definition}
A distribution $\cF_n$ over $n$-subsets of $X$ is order-statistics indistinguishable (OSI) if, for $W=\{w_1<\cdots<w_n\}\sim \cF_n$,
$$d_{TV}(W_I,W_J)\le \frac{1}{n^5}
\qquad
\forall I,J\subseteq[n],\ |I|=|J|,$$
where $W_I=\{w_i :  i \in I\}$.
\end{definition}

\citet{gravin2023onlineordinalproblemsoptimality} construct one such OSI distribution.
\begin{lemma}[Lemma 2.4 of \cite{gravin2023onlineordinalproblemsoptimality}]
\label{lem: instances existence of osi distribution}
    There exists an OSI distribution over $\mathbbm{N}$.
\end{lemma}

We construct a new OSI distribution with the added property of exponential gaps between values, by transforming the distribution of \cref{lem: instances existence of osi distribution}. 
\begin{lemma}
\label{lem: instances multiplicative gap osi distribution} 
    There exists an OSI distribution over $\largebinom{\{n^{3k}: k \in \mathbbm{N} \}}{n}$. 
\end{lemma}
\begin{proof}
    Let $f: 2^{\mathbbm{N}} \rightarrow 2^{\mathbbm{N}}$ such that for all $S \subseteq \mathbbm{N}$:
    \begin{align*}
        f(S) = \{n^{3 x}: x \in S\}
    \end{align*}
    and $I, J \subseteq [n]$ and $|I| = |J|$.
    Let $\cF_n$ be the OSI distribution of \cref{lem: instances existence of osi distribution} and $W \sim \cF_n$.
    Let $\cF'_n$ be the distribution of $W' = f(W)$, then we have:
    \begin{align*}
        d_{TV}(W'_I, W'_J) = d_{TV}(f(W_I), f(W_J)) \leq d_{TV}(W_I, W_J) \leq \frac{1}{n^5},
    \end{align*}
    where the equality is because $g(x) = n^{3 x}$ is increasing so the $i$-th order statistic of $W'$ is also the $i$-th order statistic of $W$, the first inequality is by \cref{lem: instances total variation mapping} and the second inequality is by the assumption that $\cF_n$ is OSI over $\mathbbm{N}$.
\end{proof}

Next, for any OSI distribution $\cF_n$ and  cardinal algorithm $\cardalg$, we construct  an almost-ordinal algorithm $\simu$ that simulates $\cardalg$ and satisfies the following guarantee: the probability that $\simu$ selects the highest-value item on  instance $W_{[h]} \cup \{1\}^{n-h}$, where $W \sim  \cF_n$   and $W_{[h]}$ denotes the $h$ smallest values of $W$, is, up to lower-order terms, at least the corresponding probability for $\cardalg$.  This result is similar to Theorem 1.1 in \cite{gravin2023onlineordinalproblemsoptimality} and we prove it by adapting their  construction of an ordinal algorithm  that also simulates the decisions of a cardinal algorithm (see the overview of the proof of \cref{lem: instances reduction to ordinal object ordinal algorithm} for why we cannot use that previous result directly). 

\begin{algorithm}[h]
\caption{Almost-Ordinal Simulation}\label{alg: almost ordinal simulaiton}

    \KwIn{Cardinal algorithm $\cardalg$, OSI distribution $\cF_n$ over $n$-subsets of $\mathbbm{N}$}
    \KwOut{A feasible set of accepted items}

    $\simu \gets \emptyset$ \tcp*{output}
    $\tilde{V} \gets \emptyset$ \tcp*{successfully sampled values} 
    $k \gets 0$ \tcp*{number of large items observed so far} 

    \For{timestep $t = 1$}{
        Item $\pi(1)$ arrives \\
        \If{$\pi(1)$ is a small item}{
            $\tilde{V} = \{1\}$
        }
        \Else{
            $k \gets 1$ \\
            $\pi'(1) \gets \pi(1)$ \\
            Sample $W = \{w_1 < \dots < w_n\} \sim \cF_n$ \\
            $\tilde{V} = \{w_1\}$ 
        }
        \If{$\cardalg^1(\tilde{V}, \pi[1])$ = Accept}{
            $\simu \gets \simu \cup \{\pi(1)\}$
        }
    }
    \For{timestep $t = 2, \dots, n$}{
        Item $\pi(t)$ arrives \\
        \If{$\pi(t)$ is a small item}{
            $\tilde{V} \gets \tilde{V} \cup \{1\}$
        }
        \Else{
            $k \gets k+1$ \\
            $\pi'(k) = \pi(t)$ \\
            Rank is updated to $\sigma^k$, where $\sigma^k(j) \in [k]$ is the rank of $\pi'(j)$ among $\pi'[k]$ for all $j \in [k]$ \\
            $J^k \gets \{\sigma^k(i): i \leq k - 1\}$ \\
            Sample $\tilde{W}^k = \{\tilde{w}_1 < \dots < \tilde{w}_n\} \sim \left( \cF_n \middle| \tilde{W}^k_{J^k} = \tilde{V}_{\pi'[k-1]} \right)$ \\
            \If{sampling $\tilde{W}^{k}$ fails}{
                \Return{$\emptyset$}
            }
            $\tilde{V} \gets \tilde{V} \cup \{ \tilde{w}_{\sigma^k(k)}\}$
        }

        \If{$\cardalg^t(\tilde{V}, \pi[t])$ = Accept}{
            $\simu \gets \simu \cup \{\pi(t)\}$
        }
    }

    \Return{$\simu$}

\end{algorithm}

\begin{restatable}{lemma}{ReductionToOrdinalAlgos}
    \label{lem: instances reduction to ordinal algorithms}
    Let $\cF_n$ be an OSI distribution over $\mathbbm{N}_{\geq 2}$ and let $W = \{w_1 < w_2 < \dots < w_n\} \sim \cF_n$.
    For any $\cardalg \in \cardfam$ such that $\Pr_{\pi} \left[ h \in \cardalg(W_{h} \cup \{1\}^{n-h}, \pi) \right] \geq \frac{1}{n^2}$ for all $h \in [n]$ and $W_h \in \largebinom{\{n^{3k}: k \in \mathbbm{N}\}}{h}$, there exists $\simu \in \ordfam$ such that for any $h \in [n]$:
    \begin{align*}
        \Pr_{\pi} \left[ h \in \cardalg(W_{[h]} \cup \{1\}^{n-h}, \pi) \right] 
        \leq \left(1 + \frac{1}{n^2} \right) \cdot \Pr_{\pi}\left[ h \in \simu(h, \pi) \right].
    \end{align*}
\end{restatable}

We first describe \cref{alg: almost ordinal simulaiton}, which, given a cardinal algorithm $\cardalg$ and an OSI distribution $\cF_n$ as input, constructs the desired algorithm $\simu \in \ordfam$ for this lemma. At each time step $t$, if the arriving item is small, then its value is known to be $1$. If the arriving item is large, then $\simu$ observes only its relative rank among the large items that have arrived so far. It then samples synthetic values from $\cF_n$ conditioned on two requirements. First, the synthetic values assigned to the large items seen so far must be consistent with the observed relative ranks. Second, the previously sampled synthetic values must be preserved. Thus, after the $k$-th large item arrives, the $k$ synthetic large-item values are distributed as order statistics from $\cF_n$ conditioned on the ordinal information observed so far. Finally, $\simu$ runs $\cardalg$ on the synthetic partial instance consisting of the sampled large-item values and the observed small items, and mirrors $\cardalg$'s decision for the current item. 
For the purposes of the algorithm, we denote by $\cardalg^t(V, \pi[t]) \in \{ \text{Accept, Reject}\}$ the decision of algorithm $\cardalg$ to accept or reject on timestep $t$.

\begin{proof}[Proof of~\Cref{lem: instances reduction to ordinal algorithms}.]
    Let $\simu$ be the almost-ordinal algorithm returned by  \cref{alg: almost ordinal simulaiton} over input $\cardalg$ and $\cF_n$.
    Let $\pi'$ be the order of arrival of the large items, i.e. 
    \begin{align*}
        \pi'(k) = \pi(\min\{t: k \text{ large items in } \pi{[t]}\}).
    \end{align*}
    For any set $X = \{x_1 < \dots < x_h\}$ and set $J \subseteq [h]$, we denote $X_{J} = \{x_j: j \in J\}$.
    Consider the real instance with values $W_{[h]} \cup \{1\}^{n-h}$ and the instance constructed by $\simu$ with values $\tilde{V}$.
    We start by showing that for any $h \in [n]$, we have:
    \begin{align*}
        d_{TV}(\tilde{V}_{\pi'[h]}, W_{[h]}) \leq \frac{h-1}{n^5}. \tag{1}
    \end{align*}
    We do that by inductively showing that  $d_{TV}(\tilde{V}_{\pi'[k]}, W_{[k]}) \leq \frac{k-1}{n^5}$ for $k \in [h]$. Notice that $k = h$ gives exactly (1).
    For timestep $k=1$, $d_{TV}(\tilde{V}_{\pi'[1]}, W_{[1]}) = d_{TV}(\{w_1\}, W_{[1]}) = 0$, since $w_1$ is the lowest order statistic from $W \sim \cF_n$.
    For timestep $k \in \{2, \dots, h\}$, note that $\sigma^k(j)$ is the rank of $\tilde{v}_{\pi'(j)}$ among the first $k$ large items to arrive in $\tilde{V}$, so $J^k$ holds the ranks of items $\tilde{V}_{\pi'[k-1]}$ after the $k$-th item has arrived.
    We have that:
    \begin{align*}
        d_{TV}(\tilde{V}_{\pi'[k]}, W_{[k]}) 
        &= d_{TV}(\tilde{W}^k_{J^k} \cup \{\tilde{w}_{\sigma^k(k)}\}, W_{[k]}) \\
        &= d_{TV}(\tilde{W}^k_{J^k} \cup \{\tilde{w}_{\sigma^k(k)}\}, W_{J^k} \cup \{w_{\sigma^k(k)}\}) \\
        &\leq d_{TV}(\tilde{W}^k_{J^k}, W_{J^k}) \\
        &\leq d_{TV}(\tilde{W}^k_{J^k}, W_{[k-1]}) + d_{TV}(W_{[k-1]}, W_{J^k}) \\
        &= d_{TV}(\tilde{V}_{\pi'[k-1]}, W_{[k-1]}) + d_{TV}(W_{[k-1]}, W_{J^k}) \\
        &\leq \frac{k-2}{n^5} + \frac{1}{n^5}  \\
        &= \frac{k-1}{n^5},
    \end{align*}
    where the first line is because $\tilde{V}_{\pi'[k-1]} = \tilde{W}^k_{J^k}$ by definition of $\tilde{W}^k$ and $\tilde{v}_{\pi'(k)} = \tilde{w}_{\sigma^k(k)}$ since this value is added to $\tilde{V}$ when sampling is successful,
    the second line follows as 
    $J^k \subset [k]$ and $\sigma^k(k) \in [k]$ is the rank of $\pi'(k)$.
    The third line follows by \cref{lem: instances total variation mapping}. 
    To see that, let $f(X): \binom{\mathbbm{N}}{k} \rightarrow \binom{\mathbbm{N}}{k-1}$
    be the mapping that removes the rank $\sigma^k(k)$ number from $X$, which proves that $d_{TV}(\tilde{W}^k_{J^k} \cup \{\tilde{w}_{\sigma^k(k)}\}, W_{J^k} \cup \{w_{\sigma^k(k)}\})
    \leq d_{TV}(\tilde{W}^k_{J^k}, W_{J^k})$.
    The fourth line follows by \cref{lem: instances total variation triangle}, the fifth line follows from sampling conditionally on $\tilde{W}^k_{J^k} = \tilde{V}_{\pi'[k-1]}$, the sixth line is by induction hypothesis and the assumption that $\cF_n$ is OSI.
    We have that:
    \begin{align*}
        &\Pr_{W \sim \cF_n, \pi} \left[ h \in \cardalg(W_{[h]} \cup \{1\}^{n-h}) \right] \\
        \leq\ &\Pr_{\tilde{V}, \pi} \left[ h \in \cardalg(\tilde V, \pi) \right] + d_{TV}(\cardalg(\tilde V, \pi), \cardalg(W_{[h]} \cup \{1\}^{n-h}, \pi)) \\
        =\ &\Pr_{\pi} \left[ h \in \simu(h, \pi) \right] + d_{TV}(\tilde{V}_{\pi'[h]}, W_{[h]}) \\
        \leq\ &\Pr_{\pi} \left[ h \in \simu(h, \pi) \right] + \frac{h-1}{n^5} \\
        \leq\ &\Pr_{\pi} \left[ h \in \simu(h, \pi) \right] + \frac{1}{n^4}, \tag{2}
    \end{align*}
    where the first line is by definition of total variation distance, the second line is because $\simu(h, \pi)$ makes the same decisions as $\cardalg(\tilde{V}, \pi)$ and the total variation distance is not affected by removing $n-h$ small items from each set, the third line is by (1) and the fourth line is because $h \leq n$.
    We upper bound the right term of (2) as follows:
    \begin{align*}
        \frac{1}{n^2 \cdot n^2} 
        \leq \frac{1}{n^2} \cdot \Pr_{\tilde{V}, \pi} \left[ h \in \cardalg(\tilde{V}, \pi) \right]
        = \frac{1}{n^2} \cdot \Pr_{\pi} \left[ h \in \simu(h, \pi) \right], \tag{3}
    \end{align*}
    where the inequality is by assumption $\Pr_{\pi} \left[ h \in \cardalg(W_{h} \cup \{1\}^{n-h}, \pi) \right] \geq \frac{1}{n^2}$  for all $W_{h} \in \largebinom{\{n^{3k}: k \in \mathbbm{N} \}}{h}$ and the equality is by definition of algorithm $\simu$.
    Substituting (3) into (2) gives the statement.
\end{proof}

We also show that any optimal cardinal algorithm accepts the highest value item with probability at least $1/n^2$. We do this by lower bounding with a trivial algorithm $\first$, that accepts $\pi(1)$ and rejects all other items.
We defer the proof to \cref{sec: appendix B}.
\begin{restatable}{lemma}{LemCardLowerBoundProb}
\label{lem: optimal cardinal algorithms lower bound prob}
    Let $\cardalg^* \in \cardfam$ be an optimal algorithm, then $\Pr_{\pi}\left[ h \in \cardalg^*(W_h \cup \{1\}^{n-h}, \pi) \right] \geq \frac{1}{n^2}$ for all $h \in [14]$, $W_h = \{w_1 < \dots < w_h\} \in \largebinom{\{n^{3k}: k \in \mathbbm{N} \}}{h}$.
\end{restatable}

We are now ready to prove \cref{lem: instances reduction to ordinal object ordinal algorithm}.
\begin{proof}[Proof of \cref{lem: instances reduction to ordinal object ordinal algorithm}]
    We first note that: 
    \begin{align*}
        \max_{\cardalg \in \cardfam} \min_{(s,v) \in \mathcal I} \frac{\E{\pi}{v\left( \cardalg(s, v, \pi) \right) }}{v( O(s,v))}  
        \leq \max_{\cardalg \in \cardfam} \min_{h \in \{0, \dots, 14\}} \E{V \sim \cD_h}{\frac{\E{\pi}{v\left( \cardalg(V, \pi) \right) }}{v( O(V))}}, \tag{0}
    \end{align*}
    since the corresponding instance $(s,v)$ to any $V$ in the support of $\cD_h$ is in the family of instances $\mathcal I$.
    Let $\cF_n$ be the OSI distribution over $\largebinom{\{n^{3k}: k \in \mathbbm{N} \}}{n}$ of \cref{lem: instances multiplicative gap osi distribution} and $W \sim \cF_n$.
    Also let $\cardalg^* \in \cardfam$ be the algorithm that maximizes the left-hand side of the statement and $\simu \in \ordfam$ be the algorithm given by \cref{lem: instances reduction to ordinal algorithms} that corresponds to $\cardalg^*$.
    For $h = 0$ we have:
    \begin{align*}
        \E{W \sim \cF_n}{\frac{\E{\pi}{v(\cardalg^*(W_{[0]} \cup \{1\}^{n}, \pi))}}{v(O( W_{[0]} \cup \{1\}^{n}))}}
        &= \frac{\E{\pi}{v(\cardalg^*(\{1\}^{n}, \pi))}}{v(O(\{1\}^{n}))} \\
        &= \frac{\E{\pi}{v(\cardalg^*(\{1\}^{n}, \pi))}}{n} \\
        &= \frac{\E{\pi}{v(\simu(0, \pi))}}{n}, \tag{1}
    \end{align*}
    where the first line is because $W{[0]} = \emptyset$, the second line is because $O(\{1\}^n) = [n]$ since the sizes in $I_0$ are $s_i = 1/n$ for all $i \in [n]$ and the third line is because $\simu$ does not sample any values in $I_0$, therefore its decisions are identical to those of $\cardalg^*$.
    For $h \in [14]$ we have:
    \begin{align*}
        \E{W \sim \cF_n}{\frac{\E{\pi}{v(\cardalg^*(W_{[h]} \cup \{1\}^{n-h}, \pi))}}{v(O( W_{[h]} \cup \{1\}^{n-h}))}}
        &\leq \E{W \sim \cF_n}{\Pr_{\pi}\left[ h \in \cardalg^*(W_{[h]} \cup \{1\}^{n-h}, \pi) \right] + \frac{2}{n^2}} \\
        &= \Pr_{W \sim \cF_n, \pi}\left[ h \in \cardalg^*(W_{[h]} \cup \{1\}^{n-h}, \pi) \right] + \frac{2}{n^2} \\
        &\leq \left( 1 + \frac{1}{n^2} \right) \cdot \Pr_{\pi}\left[ h \in \simu(h, \pi) \right] + \frac{2}{n^2}, \tag{2}
    \end{align*}
    where the first line follows by \cref{lem: instances reduction to ordinal objective} since $W_{[h]} \in \largebinom{\{n^{3k}: k \in \mathbbm{N}\}}{h}$, the second line is by linearity of expectation and the third line is by \cref{lem: instances reduction to ordinal algorithms} since $\Pr_{\pi}\left[ h \in \cardalg^*(W_{[h]} \cup \{1\}^{n-h}, \pi) \right] \geq \frac{1}{n^2}$ by \cref{lem: optimal cardinal algorithms lower bound prob}.
    Then we have:
    \begin{align*}
        &\min_{(s,v) \in \mathcal I} \frac{\E{\pi}{v\left( \cardalg^*(s, v, \pi) \right) }}{v( O(s,v))} \\
        \leq\ &\min_{h \in \{0, \dots, 14\}} \E{W \sim \cF_n}{\frac{\E{\pi}{v(\cardalg^*(W_{[h]} \cup \{1\}^{n-h}, \pi))}}{v(O(W_{[h]} \cup \{1\}^{n-h}))}} \\
        =\ &\min\left( \min_{h \in [14]} \E{W \sim \cF_n}{\frac{\E{\pi}{v(\cardalg^*(W_{[h]} \cup \{1\}^{n-h}, \pi))}}{v(O(W_{[h]} \cup \{1\}^{n-h}))}}, \E{W \sim \cF_n}{\frac{\E{\pi}{v(\cardalg^*(W_{[0]} \cup \{1\}^{n}, \pi))}}{v(O(W_{[0]} \cup \{1\}^{n}))}} \right)\\
        \leq\ &\left( 1 + \frac{1}{n^2} \right) \cdot \min\left( \min_{h \in [14]} \Pr_{\pi}\left[ h \in \simu(h, \pi) \right], \frac{\E{\pi}{|\simu(0, \pi)|}}{n} \right) + \frac{2}{n^2},
    \end{align*}
    where the first line follows by (0), the second line follows by splitting the min into $h \in [14]$ and $h = 0$, the third line follows by (1) and (2).
    The statement follows as $\simu \in \ordalg$, so its auxiliary objective is upper bounded by that of the almost-ordinal algorithm that maximizes it.
\end{proof}

\subsection{Properties of Almost-Ordinal Algorithms}
\label{sec: properties}

\cref{lem: instances properties of optimal algorithms} gives  properties that are satisfied by at least one  almost-ordinal algorithm that is optimal with respect to the auxiliary objective.
Restricting our search to algorithms that satisfy these properties will allow us to obtain the linear programming formulation in \cref{sec: linear program}.
\begin{lemma}
\label{lem: instances properties of optimal algorithms}
    There is an almost-ordinal algorithm $\ordalg \in \ordfam$ that maximizes the auxiliary objective $\rho(\ordalg)$ and satisfies the following conditions:
    \begin{enumerate}
        \item If $\ordalg$ accepts a small item, then no large items arrived before it,
        \item If $\ordalg$ accepts a small item, it also accepts every subsequent small item,
        \item If $\ordalg$ accepts a large item, then its value is the best so far.
    \end{enumerate}
\end{lemma}
\begin{proof}
    We start with part 1.
    Let $t$ be the arrival time of an arbitrary small item. 
    If, at time $t' < t$, a large item arrived, then the adversary has chosen an instance with $h \neq 0$ large items, so accepting at $t$ gives probability $0$ of accepting the best.
    The small item at time $t$ can therefore be rejected without loss in the auxiliary objective.

    We proceed to show part 2.
    Notice that for any optimal algorithm on these instances, once it accepts a small item, it can no longer fit large items as its remaining capacity is $1 - \frac{1}{n}$. Also, from that point on it can only gain value by accepting more small items.
    
    We proceed to show part 3. An algorithm can accept at most one large item, as that fills the capacity. Also, for instances with $h \in [14]$ large items, the objective is equal to the probability of accepting the highest value item, so  any large item that is not the best so far can therefore be rejected without loss in the auxiliary objective.
\end{proof}

\subsection{Linear Program Formulation}
\label{sec: linear program}

In this section, we use linear programming to bound, over the almost-ordinal algorithms that satisfy the properties of \Cref{sec: properties}, the optimal auxiliary objective value.
In \cref{tab:primal_dual} we give the Primal and Dual formulations and in \cref{lem: upper bound any algorithm is a feasible solution} we show that algorithms that satisfy the properties of \cref{lem: instances properties of optimal algorithms} correspond to a feasible solution of the Primal LP.
In \cref{lem: feasible dual solution} we find a feasible solution to the Dual LP, whose objective is below $1/e - 0.0003$, thus upper bounding the auxiliary objective of any almost-ordinal algorithm by weak duality.

\paragraph{Linear Program Variables.}
Let $\ordalg$ be an almost-ordinal algorithm that satisfies the properties of \cref{lem: instances properties of optimal algorithms}.
For all $t \in [n]$, $j \in \{0, \dots, 13\}$, let $\alpha^o_{t, j}$ be the probability of $\ordalg$ accepting item $\pi(t)$, conditional on $\pi(t)$ being a large item, its value being best so far and $j$ large items having arrived before $\pi(t)$:
\begin{align*}
    \alpha^o_{t, j} = \Pr_{\pi}\left[\pi(t) \in \ordalg(j+1, \pi) \middle| S_{t, j}^{j+1} \right].
\end{align*}
For all $t \in [n]$, let $\beta^o_t$ be the probability of $\ordalg$ accepting a small item for the first time on timestep $t$:
\begin{align*}
    \beta^o_t = \Pr_{\pi}\left[ \{\pi(k): k \leq t, v_{\pi(k)} = 1\} \cap \ordalg(0, \pi) = \{\pi(t)\} \right].
\end{align*}
Also let $r^o$ be the auxiliary objective of algorithm $\ordalg$:
\begin{align*}
    r^o = \rho(\ordalg).
\end{align*}
Motivated by this mapping, we use variables $r, \alpha_{t, j}, \beta_t$ for the Primal LP.

\paragraph{Linear Program Coefficients.}
For $h \in \{0, \dots 14\}, t \in [n], j \in \{0, 1, \dots, 14\}$, we define two useful events for our analysis, $S_{t, j}^h$ and $B_t^h$. 
$S_{t, j}^h$ is the event that in an instance with $h$ large items, $\pi(t)$ is a large item, is the best so far and $j$ large items arrived before $\pi(t)$:
\begin{align*}
    S_{t,j}^h
    = \{v_{\pi(t)} \neq 1, \pi(t) > \max ([h] \cap \pi[t-1]), |[h] \cap \pi[t-1]| = j\}.
\end{align*}
By \cref{lem: instances properties of optimal algorithms}, there exists an almost-ordinal algorithm with the highest auxiliary objective that only accepts large items if they are the best so far, so this event captures the relevant state for the acceptance decision of such an algorithm.
Then $\Pr_{\pi}[S_{t,j}^h] = \frac{\binom{t-1}{j} \binom{n-t}{h-j-1}}{\binom{n}{h}} \cdot \frac{1}{j+1}$.
$B_t^h$ is the event that in an instance with $h$ large items, no large items have arrived in the first $t$ arrivals:
\begin{align*}
    B_t^h
    = \{ [h] \cap \pi[t] = \emptyset \}.
\end{align*}
By \cref{lem: instances properties of optimal algorithms}, there exists an almost-ordinal algorithm with the highest auxiliary objective that only accepts small items if no large items have arrived earlier, so this event captures the relevant state for the acceptance decision of such an algorithm.
We also define four coefficients to make the linear programming formulations more concise.

Let $A_{t, j \to t', j'}$, for all $t \in [n]$, $j \in \{0, \dots, 14\}$, $t' > t$, $j' > j$, 
be the probability that in an instance with $j'+1$ large items, $\pi(t)$ is a large item, the best so far and $j$ large items arrived before $\pi(t)$, given that an item $\pi(t')$ arriving later is also large, best so far and $j'$ large items arrived before $\pi(t')$:
\begin{align*}
    A_{t,j\to t',j'} := \Pr_{\pi}[S_{t,j}^{j'+1}\mid S_{t',j'}^{j'+1}] = \frac{\binom{t-1}{j}\binom{t'-t-1}{j'-j-1}}{\binom{t'-1}{j'}} \cdot\frac{1}{j+1}.
\end{align*}

Let $D_{t \to t', j'}$, for all $t \in [n]$, $t' > t$, $j' \in \{0, \dots, 14\}$, 
be the probability that in an instance with $j'+1$ large items, none of them have arrived in the first $t$ timesteps, given that the later arriving item $\pi(t')$ is a large item, best so far and $j'$ large items have arrived before $t'$:
\begin{align*}
    D_{t \to t',j'} := \Pr_{\pi}[B_{t}^{j'+1}\mid S_{t',j'}^{j'+1}] = \frac{\binom{t'-1-t}{j'}}{\binom{t'-1}{j'}} 
\end{align*}

Let $w^h_{t, j}$, for all $t \in [n]$, $j \in \{0, \dots, 14\}$, $h \in [14]$,
be the probability that in an instance with $h$ large items, item $\pi(t)$ is the best overall and $j$ large items arrived before $\pi(t)$:
\begin{align*}
    w_{t,j}^h := \Pr_{\pi}\left[ \pi(t) = h, S^h_{t,j} \right] = \frac{1}{h} \cdot \frac{\binom{t-1}{j}\binom{n-t}{h-j-1}}{\binom{n}{h}}.
\end{align*}

Let $z_t$, for all $t \in [n]$, be the auxiliary objective we obtain by accepting small items starting with $\pi(t)$ in an instance with no large items:
\begin{align*}
    z_t := \frac{n-t+1}{n}.
\end{align*}

\begin{table}[htbp]
\centering
\caption{Primal and Dual Linear Programming Formulations}
\label{tab:primal_dual}
    \small 
    \begin{tabular}{p{0.45\textwidth} p{0.45\textwidth}}
        \toprule
        \textbf{Primal LP} & \textbf{Dual LP} \\ 
        \midrule
        $\begin{aligned}[t]
            \max\quad & r \quad & \\
            \text{s.t.} \quad & \sum_{t=1}^n \beta_t \le 1 \\
            \quad & \sum_{t=1}^n z_t \beta_t \ge r \\
            & \sum_{t=1}^n\sum_{j=0}^{h-1}w_{t,j}^h \alpha_{t,j} \ge r \quad \forall h \in [14] \\
            & \alpha_{t,j} + \sum_{t'<t}\sum_{j'<j}A_{t',j'\to t,j} \alpha_{t',j'} \\ &\qquad \qquad + \sum_{t'<t} D_{t'\to t,j} \beta_{t'}\le 1 \quad \forall t,j \\
            & r, \alpha_{t,j},\beta_t \ge 0 \quad &
        \end{aligned}$
        &
        $\begin{aligned}[t]
            \min\quad &\gamma+\sum_{t,j}\mu_{t,j}\\
            \text{s.t.} \quad & \lambda_0 + \sum_{h \in [14]} \lambda_h \ge 1 \\ 
            &\mu_{t,j} + \sum_{t'>t}\sum_{j'>j} A_{t,j\to t',j'}\mu_{t',j'} \\ &\qquad \qquad \ge \sum_{h \in [14]} w_{t,j}^h \lambda_h \quad \forall t,j \\ \\
            &\gamma + \sum_{t' > t}\sum_{j'} D_{t \to t',j'}\mu_{t',j'} \ge z_{t} \lambda_0 \quad \forall t \\ \\
            &\lambda_0,\lambda_h,\gamma,\mu_{t,j} \ge 0.
        \end{aligned}$ \\ 
        \bottomrule
    \end{tabular}
\end{table}

We now focus on showing that every almost-ordinal algorithm that satisfies the conditions of \cref{sec: properties}, corresponds to a feasible solution to the Primal LP, with its auxiliary objective being equal to the objective value of this solution.
We do this in \cref{lem: upper bound any algorithm is a feasible solution}.
Note that by \cref{lem: instances properties of optimal algorithms}, this is without loss of generality and there exists one such algorithm that achieves the highest possible auxiliary objective.
\begin{restatable}{lemma}{FeasibleSolutionAlgorithm}
\label{lem: upper bound any algorithm is a feasible solution}
    $r^o, \alpha^{o}_{t,j}, \beta^{o}_t$, is a feasible solution to the Primal LP and $r^o$ is the auxiliary objective of $\ordalg$.
\end{restatable}

The proof follows by showing that this solution satisfies the four constraints of the Primal LP and is deferred to \cref{sec: appendix B}.

We also explicitly construct a solution to the Dual LP, whose objective is lower than $1/e - 0.0003$ and verify computationally that it is feasible. 
\begin{lemma}
\label{lem: feasible dual solution}
    For $n = 500$, there exists a feasible solution $\gamma^*, \mu^*_{t,j}, \lambda^*_h$ for the Dual LP, with objective value $\gamma^* + \sum_{t, j} \mu^*_{t, j} < 1/e - 0.0003$.
\end{lemma}
\begin{proof}
    We prove the statement computationally.
    Specifically, for $n = 500$ we find values for the variables $\gamma^*, \mu^*_{t, j}, \lambda^*_h$ such that: 
    (i) all the constraints of the Dual LP are satisfied by this solution and (ii) its objective value is:
    \begin{align*}
        \gamma^* + \sum_{t, j} \mu^*_{t, j} < \frac{1}{e} - 0.0003.
    \end{align*}
    The code can be found in the anonymized github repository: \url{https://anonymous.4open.science/r/soda2027-knapsack-secretary-C558/check_feasibility.py}.
    The solution is given in three pickle files \texttt{gamma.pkl, mu.pkl, lambda.pkl}. 
    The python script \texttt{check\_feasibility.py} loads them and checks the constraints one by one. 
    The output is True only if all the constraints are satisfied.
\end{proof}

\subsection{Putting Everything Together}
\label{sec: proof of theorem 1}

We now have all the ingredients necessary for the proof of \cref{thm: Theorem 1}. 

\TheoremOne*
\begin{proof}
    Let $n = 500$ and $\gamma^*, \mu^*_{t, j}, \lambda^*_{h}$ the feasible solution for the Dual LP from \cref{lem: feasible dual solution}.
    Also let $\ordalg^*$ be an almost-ordinal algorithm that maximizes $\rho(\ordalg^*)$ and $\alpha^{o}_{t,j}, \beta^{o}_{t}, r^{o}$ be the Primal LP solution that corresponds to it.
    Then we have:
    \begin{align*}
        \rho(\ordalg^*)
        = r^{o}
        \leq \gamma^* + \sum_{t,j}\mu^*_{t,j} 
        \leq \frac{1}{e} - 0.0003, \tag{1}
    \end{align*}
    where the first inequality is because $\ordalg^*$ has auxiliary objective $r^{o}$ by \cref{lem: upper bound any algorithm is a feasible solution}, the second inequality is due to weak duality since $\alpha^{o}_{t,j}, \beta^{o}_{t}, r^{o}$ is a feasible solution for the Primal LP and $\gamma^*, \mu^*_{t, j}, \lambda^*_{h}$ is a feasible solution for the Dual LP by \cref{lem: feasible dual solution} and the third inequality is due to \cref{lem: feasible dual solution}.
    We have that:
    \begin{align*}
        &\max_{\cardalg \in \cardfam} \min_{(s, v) \in \cI} \frac{\E{\pi}{v\left( \cardalg(s, v, \pi) \right) }}{v( O((s,v)))} \\
        &\leq \left(1 + \frac{1}{n^2}\right) \cdot \max_{\ordalg \in \ordfam} \rho(\ordalg) + \frac{2}{n^2} \\
        &= \left(1 + \frac{1}{n^2}\right) \cdot \rho(\ordalg^*) + \frac{2}{n^2} \\
        &\leq \left( 1 + \frac{1}{n^2} \right) \cdot \left( \frac{1}{e} - 0.0003 \right) + \frac{2}{n^2} \\
        &< \frac{1}{e} - 0.0001,
    \end{align*}
    where the first line is due to \cref{lem: instances reduction to ordinal object ordinal algorithm}, the second line is by definition of $\ordalg^*$ as the optimal almost-ordinal algorithm, the third line is due to (1) and the fourth line follows by substituting $n = 500$.
\end{proof}

%% file: 4-Algorithm.tex
\section{The Algorithm}

In this section we present a \(0.178\)-competitive algorithm in expectation for the knapsack secretary problem. We first define a density-greedy reference solution on small items and show that, together with one additional high-value item outside this solution, it
upper-bounds \(v(\opt)\). We then describe the three-phase algorithm. In~\cref{sec:PhaseII}, we analyze the value-record phase: it gives a lower bound on the probability of selecting each of the top value ranks and shows that the event of reaching the final phase leaves the suffix unbiased. In~\cref{sec:PhaseIII}, we analyze the density-greedy phase and derive an itemwise selectability guarantee for the items of the reference solution. This already gives a $0.143$-competitive algorithm. Finally, in~\cref{section:Breaking}, we refine the analysis by exploiting the size-dependence of the Phase-III guarantee.

\begin{theorem}
Algorithm~\ref{alg:positive} is \(0.178\)-competitive in expectation for the knapsack secretary problem.
\end{theorem}

\paragraph{Continuous-time model.}
We use the following equivalent continuous-time formulation of the random-order model. The adversary fixes the set of items, together with their sizes and values. Each item \(i\) is then assigned an arrival time \(t_i\), independently and uniformly from \([0,1]\), and the algorithm observes the items in increasing order of arrival time. Since the arrival times are exchangeable and ties occur with probability zero, the induced ordering of the items is a uniformly random permutation. Conversely, any random-order instance can be viewed in this continuous-time form by independently drawing \(n\) uniform times, sorting them, and assigning the \(k\)-th smallest time to the \(k\)-th arriving item. Thus the continuous-time and random-order formulations are equivalent. This continuous-time viewpoint is standard in the secretary literature~\cite{Bruss84}.

\begin{remark}
Throughout the analysis we assume that all items have positive size. This is without loss of generality. Any item of size \(0\) and nonnegative value can be accepted immediately upon arrival without consuming capacity, and doing so can only increase the value obtained by the algorithm.  After removing such items, all remaining items have positive size, so the density \(d_i=v_i/s_i\) is well-defined.
\end{remark}

We begin by introducing some notation and terminology.  For an item \(i\), define its density by $d_i = \frac{v_i}{s_i}$. We write \(t_i\) for the arrival time of item \(i\), and $v_{\max}=\max_{i\in[n]}v_i$ for the highest value item.

\begin{definition}
An item \(i\in[n]\) is called small if \(s_i\le 1/2\), and large otherwise.
\end{definition}

Without loss of generality, we assume that no two items have the same value or density. If there are ties, we add infinitesimal noise to the value of the items to break the tie. All density orders below are with respect to decreasing density. For a set \(T\) of small items, define \(DG(T)\) to be the set of items fully selected by the density-greedy solution on \(T\):
\[
    DG(T)
    =
    \left\{
        i\in T:
        \sum_{\substack{j\in T\\ d_j>d_i}}s_j \le 1-s_i
    \right\}.
\]
Equivalently, \(i\in DG(T)\) if all higher-density items in \(T\), together with \(i\), fit into the knapsack. For each time \(t\), let $S(t)=\{i\in[n]:s_i\le1/2,\ t_i\le t\}$ be the set of small items observed by time \(t\), and write $DG(t)=DG(S(t))$. The following lemma states a key observation. Its proof is deferred to~\cref{app:lem:decomp}. 

\begin{lemma}\label{lem:decomp}
Let \(z\in\arg\max_{i\notin DG(1)}v_i\), with the convention that \(v_z=0\) if \(DG(1)=[n]\).  Then
\[
    v(\opt)\le v(DG(1))+v_z.
\]
\end{lemma}

\begin{corollary}\label{cor:decomp_vmax}
$v(\opt)\le v(DG(1))+v_{\max}.$
\end{corollary}

Our algorithm is designed to obtain value from the two terms in Lemma~\ref{lem:decomp}. It has three phases. Phase I is an exploration phase: all items are rejected, and the algorithm records the maximum value seen. Phase II is a value-record phase: an item whose value exceeds the Phase-I maximum is accepted, and the algorithm terminates. This phase is used to pay for the high-value item \(z\), and later also for the small deficits left by the final phase. If no item is selected in Phase II, then the algorithm enters Phase III. In Phase III, when a small item \(i\) arrives, the algorithm computes the current density-greedy solution \(DG(t_i)\). If \(i\in DG(t_i)\), the algorithm attempts to accept \(i\), and accepts it exactly when it fits in the remaining capacity; otherwise it rejects \(i\). The formal description is given in~\cref{alg:positive}.

\begin{algorithm}[ht]
\caption{Secretary then Small Density Greedy}\label{alg:positive}

\KwIn{Parameters \(0<a<b<1\); knapsack capacity normalized to \(1\)}
\KwOut{A feasible set of accepted items}

\(A\gets \emptyset\);

\(M\gets -\infty\)\tcp*{maximum value seen so far}

\BlankLine
\textbf{Phase I: sample.}

\For{each item \(i\) arriving at time \(t_i\le a\)}{
Reject \(i\); \(M\gets \max\{M,v_i\}\);
}

\BlankLine
\textbf{Phase II: value-record phase.}

\For{each item \(i\) arriving at time \(t_i\in(a,b]\)}{
\If{\(v_i>M\)}{
Accept \(i\);
\Return{\(\{i\}\)};
}
}

\BlankLine
\textbf{Phase III: small-item density phase.}

\If{no item arrived by time \(b\)}{
    With probability \(1-a/b\), \Return{\(\emptyset\)};
}

\(W\gets 0\)\tcp*{current accepted size in Phase III}

\For{each item \(i\) arriving at time \(t_i>b\)}{
\If{\(s_i>1/2\)}{
Reject \(i\);
}
\Else{
Let \(S(t_i)=\{j\in[n]:s_j\le1/2,\ t_j\le t_i\}\);

Compute \(DG(t_i)\), the density-greedy solution on \(S(t_i)\);

\eIf{\(i\in DG(t_i)\)}{
\eIf{\(W+s_i\le1\)}{
Accept \(i\);
\(A\gets A\cup\{i\}\);
\(W\gets W+s_i\);
}{
Reject \(i\);
}
}{
Reject \(i\);
}
}
}

\Return{\(A\)};

\end{algorithm}

\begin{remark}
The small randomization in the case where no item arrives before time \(b\) is included only for analytical cleanliness. It ensures that the event of proceeding to Phase III has probability \(a/b\) for every realized prefix set, including the empty prefix, and therefore does not bias the distribution of the suffix. We explain this decision in greater detail in \cref{sec:PhaseIII}.
\end{remark}

\subsection{Contribution of Phase II}\label{sec:PhaseII}

In this section we study two quantities: the probability that Phase II selects the maximum-value item, and the probability that the algorithm reaches Phase III. The proofs of the following can be found in~\cref{app:lem:second-phase}~and~\cref{app:lem:phase-three-probability}, respectively.

\begin{lemma}\label{lem:second-phase}
Phase II selects the maximum-value item with probability at least \(a\ln(b/a)\).
\end{lemma}

\begin{lemma}\label{lem:phase-three-probability}
The algorithm reaches Phase III with probability exactly \(a/b\).
\end{lemma}

\subsection{Contribution of Phase III}\label{sec:PhaseIII}

Throughout this section, fix an item \(i\in DG(1)\). We examine the probability that this item is selected by the algorithm. Conditional on reaching Phase III, the suffix remains unbiased. We record this formally next.

\begin{definition}
Let \(\mathcal F\) be the event that the algorithm proceeds to the suffix-processing part of Phase III.
\end{definition}

\begin{lemma}
\label{lem:phase-three-conditioning}
Conditioned on \(\mathcal F\), for every fixed item \(i\),
\[
    \Prob{t_i>b\mid \mathcal F}=1-b.
\]
Moreover, conditioned on \(t_i>b\) and \(\mathcal F\), the relative arrival time of \(i\) inside \((b,1]\) is uniform.
\end{lemma}

\begin{proof}
Let $B=\{j:t_j\le b\}$ be the prefix set. We claim that $\Prob{\mathcal F\mid B}=\frac ab$ for every realization of \(B\).

First suppose \(B\neq\emptyset\), and let \(y\) be the maximum-value item in \(B\).  Conditional on \(B\), the algorithm proceeds to Phase III if and only if Phase II selects no item. We claim that this happens if and only if \(y\) arrives in Phase I, i.e., at time at most \(a\). Indeed, if \(y\) arrives by time \(a\), then no item arriving in \((a,b]\) can beat the sample maximum, and so Phase II selects no item. Conversely, suppose that \(y\) arrives in \((a,b]\). If Phase II has not selected any item before time \(t_y\), then \(y\) beats the Phase-I sample maximum and is selected when it arrives. If Phase II has already selected an item before time \(t_y\), then Phase II has selected an item anyway. Thus whenever \(t_y\in(a,b]\), some item is selected in Phase II, and the algorithm does not proceed to Phase III.

Conditioned on the set \(B\), the arrival time of \(y\) is uniform on \([0,b]\).
Hence
\[
    \Prob{\mathcal F\mid B}
    =
    \Prob{t_y\le a\mid B}
    =
    \frac ab .
\]

If \(B=\emptyset\), then by the additional randomization step the algorithm proceeds to the suffix-processing part of Phase III with probability \(a/b\). Thus, for every realization of \(B\),
\[
    \Prob{\mathcal F\mid B}=\frac ab .
\]

Therefore conditioning on \(\mathcal F\) does not bias the identity of the prefix set \(B\), and hence does not bias the identity of the suffix set \([n]\setminus B\). Since each item belongs to the suffix independently with probability \(1-b\), it follows that, for every fixed item \(i\), $\Prob{t_i>b\mid \mathcal F}=1-b.$

Finally, conditional on the prefix set \(B\), the event \(\mathcal F\) depends only on the relative order of the items in \(B\), together with the independent coin flip in the case \(B=\emptyset\). It is independent of the arrival times and relative order of the items in the suffix. Consequently, after conditioning on \(\mathcal F\) and on the suffix set, the relative order of the suffix items remains uniformly random. In particular, conditioned on \(t_i>b\) and \(\mathcal F\), the relative arrival time of \(i\) inside \((b,1]\) is uniform.
\end{proof}

We next show that every item of \(DG(1)\) is attempted by the Phase-III rule if it arrives after Phase III has begun.

\begin{lemma}
\label{lem:dg-monotonicity}
For every item \(i\in DG(1)\), we have $i\in DG(t_i)$.
\end{lemma}

\begin{proof}
Fix \(i\in DG(1)\). Then \(i\) is small, and by the definition of \(DG(1)\),
\[
    \sum_{\substack{h:\ s_h\le1/2\\ d_h>d_i}} s_h
    \le
    1-s_i .
\]
When \(i\) arrives at time \(t_i\), the higher-density small items observed so far form a subset of all higher-density small items in the instance. Therefore
\[
    \sum_{\substack{h\in S(t_i)\\ d_h>d_i}}s_h
    \le
    1-s_i .
\]
Since \(i\in S(t_i)\), this implies \(i\in DG(t_i)\).
\end{proof}

\begin{corollary}
\label{cor:e-eligible}
Condition on \(t_i>b\) and on \(\mathcal F\). When \(i\) arrives in Phase III, the algorithm attempts to accept \(i\).
\end{corollary}

\begin{proof}
By Lemma~\ref{lem:dg-monotonicity}, we have \(i\in DG(t_i)\). Since \(t_i>b\), the item \(i\) arrives during Phase III. By the Phase-III rule, the algorithm attempts to accept every arriving small item that belongs to \(DG(t_i)\). Therefore the algorithm attempts to accept \(i\).
\end{proof}

We now bound the amount of mass that the algorithm attempts before the target item \(i\) arrives.

\begin{lemma}
\label{lem:expected-attempted-mass}
Fix \(i\in DG(1)\), and condition on \(\mathcal F\) and on \(t_i=\tau>b\). Let \(Z_\tau\) be the total size of items other than \(i\) that are attempted in Phase III before time \(\tau\). Then
\[
    \E{}{Z_\tau\mid \mathcal F,\ t_i=\tau}
    \le
    \ln\frac{\tau}{b}
\]
\end{lemma}

\begin{proof}
Let \(H\) be the set of small items other than \(i\). By Lemma~\ref{lem:phase-three-conditioning}, conditioning on \(\mathcal F\) does not bias the suffix set or the relative order of the suffix items. Therefore, after additionally conditioning on \(t_i=\tau>b\), for every \(u\in[b,\tau)\), the set
\[
    T(u)=\{j\in H:t_j\le u\}
\]
is distributed as an unbiased \(u\)-sample of \(H\).

For \(j\in H\), let \(H_{-j}=H\setminus\{j\}\), and let \(T_{-j}(u)\) be an unbiased \(u\)-sample of \(H_{-j}\).  If \(j\) arrives at time \(u\), then every other small item \(h\neq i,j\) has independently arrived before \(u\) with probability \(u\). Thus the set of small items other than \(i\) and \(j\) observed before \(j\) is distributed as \(T_{-j}(u)\). Since the algorithm attempts \(j\) precisely when \(j\) belongs to the density-greedy solution on the small items observed up to and including \(j\), we have
\[
    \Prob{j\textnormal{ is attempted at time }u}
    =
    \Prob{j\in DG(T_{-j}(u)\cup\{j\})}.
\]

Using the density of \(t_j\) on the interval \((b,\tau)\), we can write
\[
\begin{aligned}
    \E{}{Z_\tau\mid \mathcal F,\ t_i=\tau}
    &=
    \sum_{j\in H}
    s_j
    \int_b^\tau
    \Prob{j\in DG(T_{-j}(u)\cup\{j\})}\,du .
\end{aligned}
\]
It remains to bound the integrand.

Fix \(u\in[b,\tau)\), and let \(T(u)\) be an unbiased \(u\)-sample of \(H\). Since \(DG(T(u))\) is feasible,
\[
    s(DG(T(u)))\le 1
\]
always. On the other hand, by conditioning on whether \(j\in T(u)\),
\[
\begin{aligned}
    \E{}{s(DG(T(u)))}
    &=
    \sum_{j\in H}
    s_j \cdot \Prob{j\in T(u)\textnormal{ and }j\in DG(T(u))} \\
    &= \sum_{j\in H}
    s_j \cdot \Prob{j\in DG(T(u)) \mid j \in T(u)}\Prob{j \in T(u)} \\
    &=
    \sum_{j\in H}
    s_j\cdot u\cdot
    \Prob{j\in DG(T_{-j}(u)\cup\{j\})}.
\end{aligned}
\]
Therefore
\[
    \sum_{j\in H}
    s_j
    \Prob{j\in DG(T_{-j}(u)\cup\{j\})}
    =
    \frac{1}{u}\E{}{s(DG(T(u)))}
    \le
    \frac1u.
\]
Substituting this bound into the expression for \(\E{}{Z_\tau\mid \mathcal F,\ t_i=\tau}\), we obtain
\[
    \E{}{Z_\tau\mid \mathcal F,\ t_i=\tau} \le \int_b^\tau \frac{du}{u}= \ln\frac{\tau}{b}.
\]
\end{proof}

\begin{lemma}
\label{lem:conditional-acceptance}
Fix \(i\in DG(1)\), and condition on \(\mathcal F\) and on \(t_i=\tau>b\).
Then
\[
    \Prob{i\textnormal{ is accepted in Phase III}\mid \mathcal F,\ t_i=\tau}
    \ge
    1-\frac{\ln(\tau/b)}{1-s_i}.
\]
\end{lemma}

\begin{proof}
By \cref{cor:e-eligible}, the algorithm attempts to accept \(i\) when it arrives in Phase III.  Thus \(i\) can be rejected only because insufficient capacity remains.

Let \(Z_\tau\) be the total size of items other than \(i\) that are attempted in Phase III before time \(\tau\). The total accepted mass before \(i\) arrives is at most \(Z_\tau\). Hence, if $Z_\tau\le 1-s_i$, then \(i\) fits in the remaining capacity and is accepted.  Therefore
\[
    \{i\textnormal{ is rejected}\}
    \subseteq
    \{Z_\tau>1-s_i\}.
\]
By Markov's inequality and Lemma~\ref{lem:expected-attempted-mass},
\[
\begin{aligned}
    \Prob{i\textnormal{ is rejected}\mid \mathcal F,\ t_i=\tau}
    &\le
    \Prob{Z_\tau>1-s_i\mid \mathcal F,\ t_i=\tau} \\
    &\le
    \frac{\E{}{Z_\tau\mid \mathcal F,\ t_i=\tau}}{1-s_i} \\
    &\le
    \frac{\ln(\tau/b)}{1-s_i}.
\end{aligned}
\]
\end{proof}

\begin{lemma}
\label{lem:phase-three-selectability-markov}
Assume \(b\ge e^{-1/2}\).  For every \(i\in DG(1)\),
\[
    \Prob{i\textnormal{ is selected in Phase III}\mid \mathcal F}
    \ge
    3(1-b)-2\ln\frac1b .
\]
\end{lemma}

\begin{proof}
By Lemma~\ref{lem:phase-three-conditioning}, $\Prob{t_i>b\mid \mathcal F}=1-b.$ Moreover, conditioned on \(t_i>b\) and on \(\mathcal F\), the relative arrival
time of \(i\) inside \((b,1]\) is uniform. Since \(s_i\le1/2\), Lemma~\ref{lem:conditional-acceptance}
gives
\[
    \Prob{i\textnormal{ is accepted in Phase III}\mid \mathcal F,\ t_i=\tau}
    \ge
    1-2\ln\frac{\tau}{b}.
\]
The assumption \(b\ge e^{-1/2}\) ensures that this lower bound is nonnegative for all \(\tau\in[b,1]\). Therefore
\[
\begin{aligned}
    \Prob{i\textnormal{ is selected in Phase III}\mid \mathcal F}
    &\ge
    \Prob{t_i>b\mid \mathcal F}
    \cdot
    \frac{1}{1-b}
    \int_b^1
        \left(1-2\ln\frac{\tau}{b}\right)
    d\tau  \\
    &=
    \int_b^1
        \left(1-2\ln\frac{\tau}{b}\right)
    d\tau .
\end{aligned}
\]
Finally,
\[
    \int_b^1 \ln\frac{\tau}{b}\,d\tau
    =
    \ln\frac1b+b-1,
\]
and hence
\[
\begin{aligned}
    \int_b^1
        \left(1-2\ln\frac{\tau}{b}\right)
    d\tau =
    (1-b)-2\left(\ln\frac1b+b-1\right) =
    3(1-b)-2\ln\frac1b .
\end{aligned}
\]
This proves the lemma.
\end{proof}

\begin{corollary}
\label{cor:phase-three-selectability}
Assume \(b\ge e^{-1/2}\).  For every \(i\in DG(1)\),
\[
    \Prob{i\textnormal{ is selected in Phase III}}
    \ge
    \frac ab\left(3(1-b)-2\ln\frac1b\right).
\]
\end{corollary}

\begin{proof}
By Lemma~\ref{lem:phase-three-selectability-markov},
\[
    \Prob{i\textnormal{ is selected in Phase III}\mid \mathcal F}
    \ge
    3(1-b)-2\ln\frac1b .
\]
Moreover, by Lemma~\ref{lem:phase-three-probability}, $\Prob{\mathcal F}=a/b$.
Therefore
\[
\begin{aligned}
    \Prob{i\textnormal{ is selected in Phase III}}
    &\ge
    \Prob{\mathcal F}\cdot
    \Prob{i\textnormal{ is selected in Phase III}\mid \mathcal F} \\
    &\ge
    \frac ab\left(3(1-b)-2\ln\frac1b\right).
\end{aligned}
\]
\end{proof}

\begin{corollary}
\label{cor:phase-three-contribution-markov}
Assume \(b\ge e^{-1/2}\).  Then
\[
    \E{}{\textnormal{value obtained in Phase III}}
    \ge
    \frac ab\left(3(1-b)-2\ln\frac1b\right)v(DG(1)).
\]
\end{corollary}

\begin{proof}
By Corollary~\ref{cor:phase-three-selectability}, each item \(i\in DG(1)\) is
selected in Phase III with probability at least
\[
    \frac ab\left(3(1-b)-2\ln\frac1b\right).
\]
Therefore, by linearity of expectation,
\[
\begin{aligned}
    \E{}{\textnormal{value obtained in Phase III}}
    &\ge
    \sum_{i\in DG(1)}
    \frac ab\left(3(1-b)-2\ln\frac1b\right)v_i  \\
    &=
    \frac ab\left(3(1-b)-2\ln\frac1b\right)v(DG(1)).
\end{aligned}
\]
\end{proof}

This already gives a simple benchmark guarantee. Phase II selects the maximum-value item with probability at least $a\ln (b/a)$, while Phase III contributes at least $\frac ab\left(3(1-b)-2\ln\frac1b\right)$ times the value of \(DG(1)\). Combining this with Corollary~\ref{cor:decomp_vmax}, the guarantee is
\[
    \min\left\{
        a\ln\frac ba,\,
        \frac ab\left(3(1-b)-2\ln\frac1b\right)
    \right\}.
\]

Optimizing this simple benchmark over \(0<a<b<1\), with \(b\ge e^{-1/2}\), gives $b\approx 0.69009, a\approx 0.52563,$ and value approximately $0.14309.$

In the next section we refine the analysis using the size-dependence in Lemma~\ref{lem:conditional-acceptance} and combine it with a charging argument
for Phase II.

\subsection{Breaking the 1/6 Barrier}\label{section:Breaking}

The Markov-based analysis above uses only the crude fact that every item in \(DG(1)\) has size at most \(1/2\).  This loses information in two ways. First, Phase II does more than select the maximum-value item: it also selects the second, third, and later value ranks with positive probability. This creates additional ``leftover'' selectability that can be used to pay for items whose Phase III guarantee falls short of the target ratio.

Second, the Phase III guarantee itself is size-sensitive. An item of smaller size is easier to accept, because it can be blocked only after a larger amount of capacity has already been filled. Thus smaller items have better selectability in Phase III. Moreover, the density-greedy solution \(DG(1)\) has total size at most \(1\), so it cannot contain many large items. The refined analysis combines these two observations: Phase III already gives the desired guarantee for most small items, while the surplus probability from Phase II compensates for the few large items that may remain.

We first make the size-sensitivity of Phase III explicit. By Lemma~\ref{lem:conditional-acceptance}, if an item \(i\in DG(1)\) of size \(s_i\) arrives at time \(\tau>b\), then its rejection probability is controlled by $\frac{\ln(\tau/b)}{1-s_i}$. Thus an item of size at most \(s\) has a better guarantee when \(s\) is smaller.

\begin{definition}
For \(s\in[0,1/2]\), define
\[
    q_s(a,b)
    =
    \frac ab
    \left[
        (1-b)
        -
        \frac{\ln(1/b)+b-1}{1-s}
    \right].
\]
Equivalently,
\[
    q_s(a,b)
    =
    \frac ab
    \int_b^1
    \left(
        1-\frac{\ln(\tau/b)}{1-s}
    \right)d\tau .
\]
\end{definition}

\begin{lemma}\label{lem:size-sensitive-selectability}
Assume \(b\ge e^{-1/2}\).  Let \(i\in DG(1)\) have size \(s_i\le s\le 1/2\).
Then
\[
    \Prob{i\textnormal{ is selected in Phase III}}
    \ge
    q_s(a,b).
\]
\end{lemma}

\begin{proof}
Fix \(i\in DG(1)\), and condition on \(\mathcal F\) and on \(t_i=\tau>b\).
By Lemma~\ref{lem:conditional-acceptance},
\[
    \Prob{i\textnormal{ is accepted in Phase III}\mid \mathcal F,\ t_i=\tau}
    \ge
    1-\frac{\ln(\tau/b)}{1-s_i}.
\]
Since \(s_i\le s\), we have \(1-s_i\ge 1-s\), and hence
\[
    1-\frac{\ln(\tau/b)}{1-s_i}
    \ge
    1-\frac{\ln(\tau/b)}{1-s}.
\]
The assumption \(b\ge e^{-1/2}\) ensures that this lower bound is nonnegative
for all \(\tau\in[b,1]\) and all \(s\le1/2\).

By Lemma~\ref{lem:phase-three-conditioning}, conditioned on \(\mathcal F\), the
item \(i\) lies in the suffix with probability \(1-b\), and conditioned on
\(t_i>b\), its relative arrival time inside \((b,1]\) is uniform.  Therefore,
using also \(\Prob{\mathcal F}=a/b\),
\[
\begin{aligned}
    \Prob{i\textnormal{ is selected in Phase III}}
    &\ge
    \Prob{\mathcal F}\cdot \Prob{t_i>b\mid \mathcal F}
    \cdot
    \frac{1}{1-b}
    \int_b^1
    \left(
        1-\frac{\ln(\tau/b)}{1-s}
    \right)d\tau  \\
    &=
    \frac ab
    \int_b^1
    \left(
        1-\frac{\ln(\tau/b)}{1-s}
    \right)d\tau  \\
    &=
    q_s(a,b).
\end{aligned}
\]
This proves the lemma.
\end{proof}

We now explain how the refined accounting works. Fix a target competitive ratio \(\rho\). Phase III gives an item of size at most \(s\) a selection probability of at least \(q_s(a,b)\).  If \(q_s(a,b)\ge \rho\), then Phase III already gives this item the target selectability.  Otherwise, the item has a deficit: it is missing \(\rho-q_s(a,b)\) selection probability.

The purpose of the next few lemmas is to show that the total deficit over all items of \(DG(1)\) is small.  Since smaller items have better Phase III selectability, only relatively large items can create significant deficit; and since \(DG(1)\) has total size at most \(1\), there cannot be many such items. We will then use the extra value obtained by Phase II from selecting the second, third, and later value ranks to pay for these deficits.

\begin{definition}\label{def:deficit}
Fix a target value \(\rho>0\).  For \(s\in[0,1/2]\), define
\[
    \delta(s)
    =
    \left(\rho-q_s(a,b)\right)_+ .
\]
Thus \(\delta(s)\) is the amount of additional selection probability needed to raise an item of size at most \(s\) to the target \(\rho\).
\end{definition}

\begin{lemma}
\label{lem:deficit-convex}
Fix \(a,b\) and a target value \(\rho>0\). Then \(\delta\) is convex and non-decreasing on \([0,1/2]\).
\end{lemma}

\begin{proof}
For fixed \(a,b\), we can write
\[
    q_s(a,b)
    =
    C-\frac{D}{1-s},
\]
where
\[
    C=\frac ab(1-b),
    \qquad
    D=\frac ab\left(\ln\frac1b+b-1\right)>0.
\]
Hence
\[
    \rho-q_s(a,b)
    =
    \rho-C+\frac{D}{1-s}.
\]
The function \(s\mapsto 1/(1-s)\) is convex and increasing on \([0,1/2]\). Therefore \(s\mapsto \rho-q_s(a,b)\) is convex and non-decreasing on
\([0,1/2]\). Since
\[
    \delta(s)=\max\{0,\rho-q_s(a,b)\},
\]
and the maximum of two convex non-decreasing functions is again convex and non-decreasing, \(\delta\) is convex and non-decreasing on \([0,1/2]\).
\end{proof}

\begin{lemma}
\label{lem:total-deficit-bound}
Fix a target value \(\rho>0\).  If \(q_0(a,b)\ge \rho\), then for any
\(s_1,\ldots,s_r\in[0,1/2]\) satisfying
\[
    \sum_{\ell=1}^r s_\ell\le 1,
\]
we have
\[
    \sum_{\ell=1}^r \delta(s_\ell)
    \le
    2\delta(1/2).
\]
Moreover, for any single item,
\[
    \delta(s_\ell)\le \delta(1/2).
\]
\end{lemma}

\begin{proof}
Since \(\rho\le q_0(a,b)\), we have \(\delta(0)=0\).  Moreover, \(q_s(a,b)\) is non-increasing in \(s\), and therefore \(\delta(s)\) is non-decreasing. In
particular, for every \(s\in[0,1/2]\),
\[
    \delta(s)\le \delta(1/2).
\]
This proves the single-item bound.

We now prove the total deficit bound. The function \(\delta\) is convex on \([0,1/2]\). Hence, for every \(s\in[0,1/2]\),
\[
\begin{aligned}
    \delta(s)
    &=
    \delta\!\left(2s\cdot \frac12 + (1-2s)\cdot 0\right)  \\
    &\le
    2s\,\delta(1/2) + (1-2s)\delta(0) \\
    &=
    2s\,\delta(1/2).
\end{aligned}
\]
Summing over \(\ell=1,\ldots,r\), we obtain
\[
    \sum_{\ell=1}^r \delta(s_\ell)
    \le
    2\delta(1/2)\sum_{\ell=1}^r s_\ell
    \le
    2\delta(1/2),
\]
as claimed.
\end{proof}

\begin{corollary}
\label{cor:worst-case-deficit}
Fix a target value \(\rho>0\), and assume \(\rho\le q_0(a,b)\). Among all possible size profiles of \(DG(1)\), subject only to
\[
    s_i\le \frac12
    \qquad\text{for all } i\in DG(1),
    \qquad\text{and}\qquad
    \sum_{i\in DG(1)}s_i\le 1,
\]
the total Phase-III deficit is at most the deficit of two items of size \(1/2\). That is,
\[
    \sum_{i\in DG(1)} \delta(s_i)
    \le
    2\delta(1/2).
\]
In this sense, the worst case for the refined analysis is when \(DG(1)\) contains two items of size \(1/2\).
\end{corollary}

\begin{proof}
Apply Lemma~\ref{lem:total-deficit-bound} to the size profile \(\{s_i:i\in DG(1)\}\). Since \(DG(1)\) is feasible and contains only small items, we have
\[
    \sum_{i\in DG(1)}s_i\le 1
    \qquad\text{and}\qquad
    s_i\le\frac12
    \quad\text{for every }i\in DG(1).
\]
Therefore
\[
    \sum_{i\in DG(1)}\delta(s_i)
    \le
    2\delta(1/2).
\]
The right-hand side is exactly the deficit incurred by two items of size \(1/2\), which proves the claim.
\end{proof}

We also need to quantify the contribution of Phase II beyond the maximum-value item.

\begin{definition}
For \(j\ge 1\), define
\[
    P_j(a,b)
    =
    a\int_a^b \frac{(1-\tau)^{j-1}}{\tau}\,d\tau .
\]
\end{definition}

\begin{lemma}
\label{lem:phase-two-rank-probabilities}
For every \(j\ge 1\), \(P_j(a,b)\) is a lower bound on the probability that Phase II selects the \(j\)-th largest-value item.
\end{lemma}

\begin{proof}
Let \(x_j\) be the \(j\)-th largest-value item, and condition on \(t_{x_j}=\tau\in(a,b]\). We describe a sufficient event for \(x_j\) to be selected in Phase II.

First, require that all \(j-1\) items with value larger than \(v_{x_j}\) arrive after time \(\tau\). This has probability \((1-\tau)^{j-1}\). On this event, \(x_j\) has larger value than every item that arrived before time \(\tau\).

It remains to ensure that no lower-value item is selected in Phase II before time \(\tau\). Let \(m\) be the number of items with value smaller than \(v_{x_j}\). Among the lower-value items, the event that no such item arrives before \(\tau\) has probability \((1-\tau)^m\).  In that case \(x_j\) is selected when it arrives. Otherwise, consider the highest-value lower item that arrives before time \(\tau\). Conditional on the set of lower-value items that arrive before \(\tau\) being nonempty, the arrival time of this highest-value lower item is uniform on \([0,\tau]\). Hence it lies in the
initial sample with probability \(a/\tau\). If this happens, then no item arriving in \((a,\tau)\) has value larger than the Phase-I sample maximum, while \(x_j\) does have value larger than the Phase-I sample maximum.  Thus no earlier Phase-II item is selected, and \(x_j\) is selected at time \(\tau\).

Therefore, conditional on the higher-value items all arriving after \(\tau\), the probability that the lower-value items do not cause an earlier Phase-II selection is at least
\[
    (1-\tau)^m+\frac a\tau\left(1-(1-\tau)^m\right)
    \ge
    \frac a\tau .
\]
Consequently,
\[
    \Prob{\textnormal{Phase II selects }x_j\mid t_{x_j}=\tau}
    \ge
    (1-\tau)^{j-1}\cdot\frac a\tau .
\]
Integrating over \(\tau\in(a,b]\), we obtain
\[
    \Prob{x_j\textnormal{ is selected in Phase II}}
    \ge
    \int_a^b (1-\tau)^{j-1}\frac a\tau\,d\tau
    =
    P_j(a,b).
\]
\end{proof}

We are now ready to combine the two sources of value.  Let \(z\in\arg\max_{i\notin DG(1)}v_i\). By Lemma~\ref{lem:decomp}, the benchmark \(v(DG(1))+v_z\) upper-bounds \(\opt\). Phase III almost pays for the \(DG(1)\) part: it gives each item \(i\in DG(1)\) selection probability \(\rho\), except for its deficit \(\delta(s_i)\). Phase II will pay for the benchmark item \(z\), with coefficient \(\rho\), and for the Phase-III deficits of the items in \(DG(1)\). The total-deficit bound says that these deficits are no worse than two deficits of size-\(1/2\) items. The proof of the following is deferred to~\cref{app:lem:charging-total-size-deficits}.

\begin{lemma}
\label{lem:charging-total-size-deficits}
Fix a target value \(\rho>0\), and assume \(b\ge e^{-1/2}\) and
\(\rho\le q_0(a,b)\).  Let
\[
    \delta_{1/2}
    =
    \left(\rho-q_{1/2}(a,b)\right)_+ .
\]
Suppose that
\[
    P_1(a,b)\ge \rho,
\]
\[
    P_1(a,b)+P_2(a,b)\ge \rho+\delta_{1/2},
\]
and
\[
    P_1(a,b)+P_2(a,b)+P_3(a,b)\ge \rho+2\delta_{1/2}.
\]
Then the expected contribution of Phase II and Phase III is at least
\[
    \rho\left(v(DG(1))+v_z\right),
\]
where \(z\in\arg\max_{i\notin DG(1)}v_i\), with the convention that \(v_z=0\) if \(DG(1)=[n]\).
\end{lemma}


\paragraph{Putting Everything Together.}

We now instantiate the refined analysis. Let $\Gamma_b=\ln\frac1b+b-1.$ Recall that
\[
    q_s(a,b)
    =
    \frac ab
    \left[
        (1-b)-\frac{\Gamma_b}{1-s}
    \right].
\]
In particular,
\[
    q_0(a,b)
    =
    \frac ab\left((1-b)-\Gamma_b\right),
\]
and
\[
    q_{1/2}(a,b)
    =
    \frac ab\left((1-b)-2\Gamma_b\right).
\]

We set the target competitive ratio to be $\rho=q_0(a,b)$. Then the condition \(\rho\le q_0(a,b)\) required by Lemma~\ref{lem:total-deficit-bound} holds with equality. The deficit of a size-\(1/2\) item is
\[
\begin{aligned}
    \delta_{1/2} =    \left(\rho-q_{1/2}(a,b)\right)_+  = q_0(a,b)-q_{1/2}(a,b) =
    \frac ab\,\Gamma_b .
\end{aligned}
\]
By Lemma~\ref{lem:total-deficit-bound}, the total Phase-III deficit over all items of \(DG(1)\) is at most \(2\delta_{1/2}\), and the deficit of any single
item is at most \(\delta_{1/2}\). Therefore, by Lemma~\ref{lem:charging-total-size-deficits}, it suffices to verify
\[
    P_1(a,b)\ge \rho,
\]
\[
    P_1(a,b)+P_2(a,b)\ge \rho+\delta_{1/2},
\]
and
\[
    P_1(a,b)+P_2(a,b)+P_3(a,b)
    \ge
    \rho+2\delta_{1/2}.
\]

We choose
\[
    b=0.61867,
    \qquad
    a=0.39189.
\]
For these parameters,
\[
    \rho
    =
    q_0(a,b)
    \approx
    0.17893191,
\]
and
\[
    \delta_{1/2}
    =
    \frac ab\,\Gamma_b
    \approx
    0.06261756.
\]
Moreover,
\[
    q_{1/2}(a,b)
    =
    \rho-\delta_{1/2}
    \approx
    0.11631435.
\]

It remains to verify the three charging inequalities.  Using the formulas
\[
    P_1(a,b)=a\ln\frac ba,
\]
\[
    P_2(a,b)
    =
    a\left(\ln\frac ba-(b-a)\right),
\]
and
\[
    P_3(a,b)
    =
    a\left(\ln\frac ba-2(b-a)+\frac{b^2-a^2}{2}\right),
\]
we obtain
\[
    P_1(a,b)
    \approx
    0.17893338
    >
    \rho
    \approx
    0.17893191,
\]
\[
    P_1(a,b)+P_2(a,b)
    \approx
    0.26899394
    >
    \rho+\delta_{1/2}
    \approx
    0.24154948,
\]
and
\[
    P_1(a,b)+P_2(a,b)+P_3(a,b)
    \approx
    0.31508735
    >
    \rho+2\delta_{1/2}
    \approx
    0.30416704.
\]
Thus the conditions of Lemma~\ref{lem:charging-total-size-deficits} are satisfied. Therefore
\[
    \E{}{\textnormal{Phase II}+\textnormal{Phase III}}
    \ge
    \rho\left(v(DG(1))+v_z\right).
\]
By Lemma~\ref{lem:decomp},
\[
    v(\opt)\le v(DG(1))+v_z.
\]
Hence
\[
    \E{}{\alg}
    \ge
    \rho\,v(\opt)
    \ge
    0.1789\,v(\opt).
\]
Thus Algorithm~\ref{alg:positive} is \(0.1789\)-competitive in expectation.

%% file: A-AbelsImpossibility.tex
\section{Discussion of  Impossibility Result for Ordinal Algorithms from Prior Work}
\label{app:previouswork}

An algorithm is value-ordinal if it does not observe the cardinal values of the items but instead only observes the relative order of item values.  \citet{abels2022knapsack} show that value-ordinal algorithms for knapsack secretary cannot achieve competitive ratio better than $1/(e+1)+o(1),$ even for $1$--$B$ instances,  which are instances where all the items either have size $1/B$ or $1$, for some large $B$.  We note that the family of hard instances we construct are also $1$--$B$ instances.
Their proof also uses linear programming to find the optimal ordinal algorithm on a different family of  hard $1$--$B$ instances.

If one could show that there is a sufficiently small gap between the performance of an optimal cardinal algorithm and the optimal value-ordinal algorithm for their family of $1$--$B$ instances, then, combined with their $1/(e+1)+o(1)$ impossibility for value-ordinal algorithms, it would provide a proof of our main impossibility result. However, as shown below, there exists a cardinal $1/2$-competitive algorithm for their family of hard instances. In contrast, we show that there is an asymptotically negligible gap between the optimal cardinal algorithm and the optimal almost-ordinal algorithm (see Definition~\ref{def: almost ordinal algorithm} for the definition of almost-ordinal) for our  family of hard instances.

The $1$--$B$ instances  used in \citet{abels2022knapsack} have  $n = 2B$ items.
They call items $\{1, \dots, B\}$ \textit{large} (size $1$ items) and items $\{B+1, \dots, 2B\}$ \textit{small} (size $1/B$ items).
They consider two instances, that differ only in the value of item $1$:
\begin{align*}
    I_1: v_i =
    \begin{cases}
        1 + (B - i) \cdot \epsilon, &\text{ if } i \leq B,\\
        1 - i \cdot \epsilon, &\text{ otherwise}
    \end{cases}
    \qquad \qquad
    I_2: v_i = 
    \begin{cases}
        B^2, &\text{ if } i = 1, \\
        1 + (B - i) \cdot \epsilon, &\text{ if } 2 \leq i \leq B, \\
        1 - i \cdot \epsilon, &\text{ otherwise}
    \end{cases}
\end{align*}
These two instances do not suffice for extending the impossibility result to cardinal algorithms. 
Indeed, there is a cardinal algorithm with ratio $1/2$ for these instances that works in two phases: 
\begin{itemize}
    \item \textbf{Phase I:} For the first $B$ arrivals, only accept an item with value $B^2$ and reject all others,
    \item \textbf{Phase II:} For the last $B$ arrivals, if no item was accepted in the first phase, accept all items with value less than $1$, otherwise, reject all items.
\end{itemize}
In instance $I_1$, the optimal solution is $\{B+1, \dots, 2B\}$, with value $\sum_{i=B+1}^{2B}v_i$. 
In this instance, the algorithm accepts any small items that arrive in the second half, therefore, it accepts each small item with probability $1/2$. By linearity of expectation, this gives value $1/2 \cdot \sum_{i=B+1}^{2B}v_i$, so the ratio is $1/2$.
It is also easy to see that this algorithm has ratio at least $1/2$ on instance $I_2$, where the optimal solution is $\{1\}$ with value $B^2$. 
Item $1$ arrives in the first half with probability $1/2$, in which case the algorithm always accepts it.

%% file: B-MissingProofsUpperBound.tex
\section{Omitted Proofs From the Upper Bound}
\label{sec: appendix B}

\subsection{The proof of \cref{lem: instances total variation triangle,lem: instances total variation mapping}}

\TVTriangleInequality*
\begin{proof}
    \begin{align*}
        d_{TV}(X, Y) 
        &= \frac{1}{2} \cdot \sum_{t \in \cT} | \Pr_{X}\left[ X = t \right] - \Pr_{Y}\left[ Y = t \right] | \\
        &= \frac{1}{2} \cdot \sum_{t \in \cT} | \left( \Pr_{X}\left[ X = t \right] - \Pr_{Z}\left[ Z = t \right] \right) + \left( \Pr_{Z}\left[ Z = t \right] - \Pr_{Y}\left[ Y = t \right] \right) | \\
        &\leq \frac{1}{2} \cdot \sum_{t \in \cT} | \Pr_{X}\left[ X = t \right] - \Pr_{Z}\left[ Z = t \right] | + \frac{1}{2} \sum_{t \in \cT} | \Pr_{Z}\left[ Z = t \right] - \Pr_{Y}\left[ Y = t \right] |\\
        &= d_{TV}(X, Z) + d_{TV}(Z, Y),
    \end{align*}
    where the first line is by definition, the second line is by adding and subtracting $\Pr_{Z}\left[ Z = t \right]$, the third line is by triangle inequality of the absolute value and the fourth line is by definition.
\end{proof}

\TVMapping*
\begin{proof}
    The marginal distribution of $f(X)$ is:
    \begin{align*}
        \Pr_{X, f}\left[ f(X) = r \right]
        = \sum_{t \in \cT} \Pr_{X, f}\left[X = t, f(t) = r \right] 
        = \sum_{t \in \cT} \Pr_{X}\left[X = t \right] \cdot \Pr_{f}\left[ f(t) = r \right], \tag{*}
    \end{align*}
    where the second equality is by independence of $f, X$.
    Similarly the marginal distribution of $f(Y)$ is $\Pr_{Y, f}\left[ f(Y) = r \right] = \sum_{t \in \cT} \Pr_{Y}\left[Y = t \right] \cdot \Pr_{f}\left[ f(t) = r \right]$.
    Then we have:
    \begin{align*}
        d_{TV}(f(X), f(Y)) 
        &= \frac{1}{2} \cdot \sum_{r \in \cR} | \Pr_{X, f}\left[ f(X) = r \right] - \Pr_{Y, f}\left[ f(Y) = r \right] | \\
        &= \frac{1}{2} \cdot \sum_{r \in \cR} | \sum_{t \in \cT} \Pr_{f}\left[ f(t) = r \right] \cdot \left( \Pr_{X}\left[ X = t \right] -  \Pr_{Y}\left[ Y = t \right] \right) | \\
        &\leq \frac{1}{2} \cdot \sum_{r \in \cR} \sum_{t \in \cT} \Pr_{f}\left[ f(t) = r \right] \cdot | \Pr_{X}\left[ X = t \right] -  \Pr_{Y}\left[ Y = t \right] | \\
        &\leq \frac{1}{2} \cdot \sum_{t \in \cT} | \Pr_{X}\left[ X = t \right] -  \Pr_{Y}\left[ Y = t \right] | \cdot \sum_{r \in \cR} \Pr_{f}\left[ f(t) = r \right] \\
        &\leq \frac{1}{2} \cdot \sum_{t \in \cT} | \Pr_{X}\left[ X = t \right] -  \Pr_{Y}\left[ Y = t \right] | \\
        &= d_{TV}(X, Y),
    \end{align*}
    where the first line is by definition, the second line is due to (*), the third line is by triangle inequality of the absolute value, the fourth line is by changing the order of summation, the fifth line is because $\sum_{r \in \cR}\Pr_{f}\left[ f(t) = r \right] = 1$ for all $t \in \cT$ and the sixth line is by definition.
\end{proof}

\subsection{The proof of \cref{lem: optimal cardinal algorithms lower bound prob}}

\LemCardLowerBoundProb*
\begin{proof}
    We have that for any $h \in [14]$:
    \begin{align*}
        \frac{\E{\pi}{v(\cardalg^*(W_{h} \cup \{1\}^{n-h}, \pi))}}{v(O( W_{h} \cup \{1\}^{n-h}))}
        &\leq \Pr_{\pi}\left[ h \in \cardalg^*(W_{h} \cup \{1\}^{n-h}, \pi) \right] + \frac{2}{n^2}  \tag{1}
    \end{align*}
    where the inequality follows by \cref{lem: instances reduction to ordinal objective} since $W_{h} \in \largebinom{\{n^{3k}: k \in \mathbbm{N}\}}{h}$.
    Notice that we can define a trivial algorithm $\first \in \cardfam$ that always accepts $\pi(1)$ and rejects everything else, and it is well-defined since $s_i \leq 1$ for all $i \in [n]$. 
    Then we have:
    \begin{align*}
        &\min_{h \in \{0, \dots, 14\}} \frac{\E{\pi}{\cardalg^*(W_h \cup \{1\}^{n-h})}}{v(O(W_h \cup \{1\}^{n-h}))} \\
        \geq\ &\min_{h \in \{0, \dots, 14\}} \frac{\E{\pi}{\first(W_{h} \cup \{1\}^{n-h})}}{v(O(W_{h} \cup \{1\}^{n-h}))}\\
        =\ &\min\left( \frac{1}{n}, \min_{h \in [14]} \frac{\E{\pi}{\first(W_{h} \cup \{1\}^{n-h}, \pi)}}{w_h} \right) \\
        =\ &\min\left( \frac{1}{n}, \min_{h \in [14]} \frac{1}{n} \cdot \frac{w_h + \sum_{i=1}^{h-1}w_i + n-h}{w_h} \right) \\
        =\ &\frac{1}{n}, \tag{2}
    \end{align*}
    where the first line is by optimality of $\cardalg^*$, the second line follows by splitting cases $h=0$, where $v(O(W_{h} \cup \{1\}^{n-h})) = n$ and $h \in [14]$ where $v(O(W_{h} \cup \{1\}^{n-h})) = w_h$, the third line follows because $\first$ accepts only $\pi(1)$ and gets expected value $(\sum_{i=1}^{h}w_h + n-h) / n$ and the fourth line because $(\sum_{i=1}^{h}w_h + n-h) / n \geq 1$.
    By substituting (2) in (1) we have:
    \begin{align*}
        \Pr_{\pi}\left[ h \in \cardalg^*(W_{h} \cup \{1\}^{n-h}, \pi) \right]
        \geq \frac{1}{n} - \frac{2}{n^2}
        = \frac{n - 2}{n^2}
        \geq \frac{1}{n^2},
    \end{align*}
    where the last equality is due to $n \geq 14$ in family $\cI$.
\end{proof}

\subsection{The proof of \cref{lem: upper bound any algorithm is a feasible solution}}

\FeasibleSolutionAlgorithm*
The proof follows by showing that this solution satisfies the four constraints of the Primal LP and we do that in \cref{lem: first constraint,lem: second constraint,lem: third constraint,lem: fourth constraint}.

\begin{restatable}{lemma}{LemFirstConstraint}
\label{lem: first constraint}
    $\sum_{t=1}^{n}\beta^o_t \leq 1$.
\end{restatable}
\begin{proof}
    We have that:
    \begin{align*}
        \sum_{t=1}^{n} \beta^{o}_t 
        &= \sum_{t=1}^{n} \Pr_{\pi}\left[ \{\pi(k): k \leq t, v_{\pi(k)} = 1\} \cap \ordalg(0, \pi) = \{\pi(t)\} \right] \\
        &= \Pr_{\pi}\left[ \{t: v_t = 1\} \cap \ordalg(0, \pi) \neq \emptyset \right] \\
        &\leq 1,
    \end{align*}
    where the first equality is by definition of $\beta^{o}_t$, the second equality is by law of total probability, and the inequality by definition of probability.
\end{proof}

\begin{restatable}{lemma}{LemSecondConstraint}
\label{lem: second constraint}
    $\sum_{t=1}^{n}z_t \cdot \beta^o_t \geq r^o$.
\end{restatable}
\begin{proof}
    Since $r^o$ is the auxiliary objective of $\ordalg$, we have:
    \begin{align*}
        r^o \leq \rho(\ordalg) \leq \frac{1}{n} \cdot \E{\pi}{|\ordalg(\{1\}^n, \pi)|}. \tag{1}
    \end{align*}
    We have:
    \begin{align*}
        \E{\pi}{ |\ordalg(0, \pi)| } 
        =\ &\sum_{t=1}^{n} (n-t+1) \cdot \Pr\left[ \{\pi(k): k \leq t, v_{\pi(k)} = 1\} \cap \ordalg(0, \pi) = \{\pi(t)\} \right] \\
        =\ &\sum_{t=1}^{n} (n-t+1) \cdot \beta^{o}_{t}, \tag{2}
    \end{align*}
    where the first line is because when $\ordalg$ starts accepting small items at timestep $t$, then $\ordalg$ must accept all the next $1$s as well by \cref{lem: instances properties of optimal algorithms}, and the second line is by definition of $\beta^{o}_t$. 
    By substituting (2) into (1) and dividing by $n$ we show that:
    \begin{align*}
        \sum_{t=1}^{n} \frac{n-t+1}{n} \cdot \beta^{o}_{t} \geq r^{o}.
    \end{align*}
\end{proof}
    
\begin{restatable}{lemma}{LemThirdConstraint}
\label{lem: third constraint}
    $\sum_{t=1}^{n}\sum_{j=0}^{h-1} w^h_{t, j} \cdot \alpha^o_{t, j} \geq r^o$ for all $h \in [14]$.
\end{restatable}
\begin{proof}
    Let $h \in [14]$.
    Achieving auxiliary objective $r^{o}$ in this instance is equivalent to:
    \begin{align*}
        r^o \leq \rho(\ordalg) \leq \Pr_{\pi}\left[ h \in \ordalg(h, \pi) \right]. \tag{1}
    \end{align*}
    We have that:
    \begin{align*}
        \Pr_{\pi}\left[ h \in \ordalg(h, \pi) \right]
        =\ &\sum_{t=1}^{n} \sum_{j=0}^{h-1} \Pr_{\pi}\left[ \pi(t) \in \ordalg(h, \pi), \pi(t) = h, S^h_{t,j} \right] \\
        =\ &\sum_{t=1}^{n} \sum_{j=0}^{h-1} \Pr_{\pi}\left[ \pi(t) = h, S^h_{t,j} \right] \cdot \Pr_{\pi}\left[ \pi(t) \in \ordalg(h, \pi) \middle| \pi(t) = h, S^h_{t,j} \right] \\
        =\ &\sum_{t=1}^{n} \sum_{j=0}^{h-1} \Pr_{\pi}\left[ \pi(t) = h, S^h_{t,j} \right] \cdot \Pr_{\pi}\left[ \pi(t) \in \ordalg(j+1, \pi) \middle| S^{j+1}_{t,j} \right] \\
        =\ &\sum_{t=1}^{n} \sum_{j=0}^{h-1} w^h_{t, j} \cdot \alpha^{o}_{t, j}, \tag{2}
    \end{align*}
    where the first line is by law of total probability, the second line is by Bayes' rule, the third line is because $\ordalg$ accepts $\pi(t)$ only based on information it has observed, so it cannot differentiate between $\pi(t)$ being best so far and best overall or whether the total number of large items in this instance is $j+1$ or $h$, and the fourth line is by definition of $w^h_{t,j}, \alpha^o_{t,j}$.
    By substituting (2) into (1) we show that:
    \begin{align*}
        \sum_{t=1}^{n} \sum_{j=0}^{h-1} w^h_{t, j} \cdot \alpha^{o}_{t, j} \geq r^{o} \quad \forall h \in [14].
    \end{align*}
\end{proof}

\begin{restatable}{lemma}{LemFourthConstraint}
\label{lem: fourth constraint}
    $\alpha^o_{t,j}  + \sum_{t' = 1}^{t-1} D_{t' \rightarrow t, j} \cdot \beta^{o}_{t'} + \sum_{t' = 1}^{t-1} \sum_{j' = 0}^{j-1} A_{t', j' \rightarrow t, j} \cdot \alpha^{o}_{t', j'}  \leq 1$ for all $t \in [n]$, $j \in \{0, \ldots, 13\}$.
\end{restatable}
\begin{proof}
    We have that:
    \begin{align*}
        \alpha^o_{t, j}
        =\ &\Pr_{\pi}\left[ \pi(t) \in \ordalg(j+1, \pi) \middle| S^{j+1}_{t, j}   \right] \\
        =\ &\Pr_{\pi}\left[ \pi[t] \cap \ordalg(j+1, \pi) = \{\pi(t)\} \middle| S^{j+1}_{t, j}   \right] \\
        \leq\ &\Pr_{\pi}\left[ \pi[t-1] \cap \ordalg(j+1, \pi) = \emptyset \middle| S^{j+1}_{t,j}   \right] \\
        =\ &1 - \Pr_{\pi}\left[ \{k \in \pi[t-1]: v_k = 1\} \cap \ordalg(j+1, \pi) \neq \emptyset  \middle| S^{j+1}_{t,j}   \right] 
        \\ &\qquad \qquad - \Pr_{\pi}\left[ \{k \in \pi[t-1]: v_k \neq 1\} \cap \ordalg(j+1, \pi) \neq \emptyset \middle| S^{j+1}_{t, j}   \right] \\
        =\ &1 - \sum_{t' = 1}^{t-1} \Pr_{\pi}\left[ \min\{k \in [n]: v_{\pi(k)} = 1, \pi(k) \in \ordalg(j+1, \pi) \} = t' \middle| S^{j+1}_{t,j}   \right]  
        \\ &\qquad \qquad - \sum_{t' = 1}^{t-1} \sum_{j'=0}^{j-1} \Pr_{\pi}\left[ \pi(t') \in \ordalg(j+1, \pi), S^{j+1}_{t', j'} \middle| S^{j+1}_{t,j}   \right] \\
        =\ &1 - \sum_{t' = 1}^{t-1} \Pr_{\pi}\left[ \min\{k \in [n]: v_{\pi(k)} = 1, \pi(k) \in \ordalg(j+1, \pi) \} = t', B^{j+1}_{t'} \middle| S^{j+1}_{t,j}   \right]
        \\ &\qquad \qquad - \sum_{t' = 1}^{t-1} \sum_{j' = 0}^{j-1} \Pr_{\pi}\left[ \pi(t') \in \ordalg(j+1, \pi), S^{j+1}_{t', j'} \middle| S^{j+1}_{t,j}   \right]\\
        =\ &1 - \sum_{t' = 1}^{t-1} D_{t' \rightarrow t, j} \cdot \Pr_{\pi}\left[ \min\{k \in [n]: v_{\pi(k)} = 1, \pi(k) \in \ordalg(j+1, \pi) \} = t' \middle| B^{j+1}_{t'}   \right] 
        \\ &\qquad \qquad - \sum_{t' = 1}^{t-1} \sum_{j' = 0}^{j-1} A_{t', j' \rightarrow t, j} \cdot \Pr_{\pi}\left[ \pi(t') \in \ordalg(j+1, \pi)  \middle| S^{j+1}_{t', j'} \right]\\
        =\ &1 - \sum_{t' = 1}^{t-1} D_{t' \rightarrow t, j} \cdot \beta^{\pi}_{t'} - \sum_{t' = 1}^{t-1} \sum_{j' = 0}^{j-1} A_{t', j' \rightarrow t, j} \cdot \alpha^{\pi}_{t', j'},
    \end{align*}
    where the first line is by definition of $\alpha^{o}_{t,j}$, the second line is because $\ordalg$ can only accept large item $\pi(t)$ if it has rejected all previous arrivals by construction of the sizes, the third line follows by upper bounding the probability of $\ordalg$ accepting $\pi(t)$ with the probability that $\pi(t)$ is available, the fourth line follows because the event that we accept some small items and the event that we accept a large item are mutually exclusive by construction of the sizes, the fifth line is by law of total probability, the sixth line is because $\ordalg$ only accepts a small item if no  large item had arrived earlier by \cref{lem: instances properties of optimal algorithms}, the seventh line is because $D_{t' \rightarrow t, j}$ is the probability that there are no large items before $\pi(t')$ conditional on $\pi(t)$ being best so far and $j$ large items before $\pi(t)$, $A_{t', j' \rightarrow t, j}$ is the probability that $\pi(t')$ is best so far with $j'$ large items before $\pi(t')$ conditional on $\pi(t)$ being best so far and $j$ large items before $\pi(t)$ and the fact that $\ordalg$ only makes decisions based on the past so $\pi(t)$ being best so far is independent.
\end{proof}

%% file: C-MissingProofsAlgorithm.tex
\section{Omitted Proofs From Algorithm Analysis}

\subsection{The Proof of Lemma~\ref{lem:decomp}}
\label{app:lem:decomp}

Let \(O\) be an optimal solution. We split \(O\) into its small and large items:
\[
    O_S=\{i\in O:s_i\le1/2\},
    \qquad
    O_L=\{i\in O:s_i>1/2\}.
\]

We first consider the case \(O_L=\emptyset\). Then \(O\) consists only of small items. The fractional density-greedy solution on the small items upper-bounds the value of any feasible solution using only small items. The integral density-greedy solution \(DG(1)\) differs from this fractional solution by at most one boundary item. This boundary item, if it exists, is not in
\(DG(1)\), and therefore has value at most \(v_z\).  Hence
\[
    v(O)\le v(DG(1))+v_z.
\]

It remains to consider the case \(O_L\neq\emptyset\).  Since every large item has size strictly larger than \(1/2\), the set \(O_L\) contains exactly one item. This item is not in \(DG(1)\), because \(DG(1)\) contains only small items. Therefore $v(O_L)\le v_z$. Moreover, the remaining small items in \(O\) have total size strictly less than \(1/2\), i.e., $s(O_S)<\frac12$.

We claim that \(v(O_S)\le v(DG(1))\). If \(s(DG(1))<1/2\), then \(DG(1)\) contains all small items: otherwise, since every remaining small item has size at most \(1/2\), the next small item in density order would still fit. The claim is then immediate.

Otherwise, \(s(DG(1))\ge1/2\). Since \(DG(1)\) consists of the highest-density small items up to total size at least \(1/2\), its value is at least the value of the fractional density-greedy solution of capacity \(1/2\) on the small items. This fractional solution upper-bounds every small-item solution of total size at most \(1/2\), and in particular upper-bounds \(O_S\). Thus $v(O_S)\le v(DG(1)).$

Combining the two bounds gives
\[
    v(O)
    =
    v(O_S)+v(O_L)
    \le
    v(DG(1))+v_z.
\]

\subsection{The Proof of Lemma~\ref{lem:second-phase}}
\label{app:lem:second-phase}

Let \(x\) be the maximum-value item in the whole instance, and condition on \(t_x=\tau\in(a,b]\).  Since \(x\) is the maximum-value item overall, Phase II selects \(x\) if no earlier Phase-II item is selected before \(x\) arrives.

A sufficient condition for this is that the maximum-value item among those arriving before time \(\tau\), if such an item exists, lies in the sample interval \([0,a]\).  Conditional on \(t_x=\tau\), the relative order of the items arriving before \(\tau\) is uniform, and therefore this maximum-value earlier item lies in \([0,a]\) with probability \(a/\tau\).  If no item arrives before \(\tau\), then \(x\) is selected as well, so this only increases the probability. Hence
\[
    \Prob{\textnormal{Phase II selects }x\mid t_x=\tau}
    \ge
    \frac a\tau .
\]

Since \(t_x\) is uniform on \([0,1]\), we integrate over
\(\tau\in(a,b]\):
\[
\begin{aligned}
    \Prob{\textnormal{Phase II selects }x}
    \ge
    \int_a^b \frac a\tau\,d\tau  =
    a\ln\frac ba .
\end{aligned}
\]
This proves the lemma.

\subsection{The Proof of Lemma~\ref{lem:phase-three-probability}}
\label{app:lem:phase-three-probability}

Let $B=\{i:t_i\le b\}$ be the set of items revealed by time \(b\).

First suppose \(B\neq\varnothing\), and let \(y\) be the maximum-value item in \(B\). In this case, the algorithm reaches Phase III if and only if Phase II selects no item. We claim that this happens if and only if \(y\) arrives in Phase I, i.e., at time at most \(a\).  Indeed, if \(y\) arrives by time \(a\), then no item arriving in \((a,b]\) can beat the sample maximum, and so Phase II selects no item. Conversely, suppose that \(y\) arrives in \((a,b]\). If Phase II has not selected any item before time \(t_y\), then \(y\) beats the Phase-I sample maximum and is selected when it arrives.  If Phase II has already selected an item before time \(t_y\), then Phase II selects an item anyway. Thus whenever \(t_y\in(a,b]\), some item is selected in Phase II, and the algorithm does not reach Phase III.

Conditioned on the set \(B\), the arrival time of \(y\) is uniform on \([0,b]\). Hence
\[
    \Prob{\textnormal{the algorithm reaches Phase III}\mid B}
    =
    \Prob{t_y\le a\mid B}
    =
    \frac ab .
\]

If \(B=\emptyset\), then by the additional randomization step the algorithm reaches Phase III with probability \(a/b\). Therefore, for every realization of \(B\),
\[
    \Prob{\textnormal{the algorithm reaches Phase III}\mid B}
    =
    \frac ab .
\]
Taking expectations over \(B\), the algorithm reaches Phase III with probability exactly \(a/b\).

\subsection{The Proof of Lemma~\ref{lem:charging-total-size-deficits}}
\label{app:lem:charging-total-size-deficits}

For every item \(i\in DG(1)\), define its Phase-III deficit by
\[
    \delta_i
    =
    \left(\rho-q_{s_i}(a,b)\right)_+ .
\]
By Lemma~\ref{lem:size-sensitive-selectability}, Phase III selects item \(i\) with probability at least \(q_{s_i}(a,b)\). Equivalently,
\[
    q_{s_i}(a,b)\ge \rho-\delta_i.
\]
Therefore the expected value obtained in Phase III is at least
\[
    \sum_{i\in DG(1)}(\rho-\delta_i)v_i
    =
    \rho\,v(DG(1))-\sum_{i\in DG(1)}\delta_i v_i .
\]

It remains to show that Phase II contributes enough value to pay for
\[
    \rho v_z+\sum_{i\in DG(1)}\delta_i v_i .
\]
By Lemma~\ref{lem:total-deficit-bound}, applied to the size profile of \(DG(1)\), we have
\[
    \sum_{i\in DG(1)}\delta_i\le 2\delta_{1/2},
\]
and, for every \(i\in DG(1)\),
\[
    \delta_i\le \delta_{1/2}.
\]
Moreover, since \(b\ge e^{-1/2}\), the quantities \(q_s(a,b)\) are nonnegative for all \(s\in[0,1/2]\), and hence \(\delta_i\le\rho\). Thus the coefficient vector consisting of the coefficient \(\rho\) for \(z\) and the coefficients \(\{\delta_i:i\in DG(1)\}\), when sorted in non-increasing order, has first prefix sum at most \(\rho\), second prefix sum at most \(\rho+\delta_{1/2}\), and every larger prefix sum at most \(\rho+2\delta_{1/2}\).

Let $V_1\ge V_2\ge \cdots \ge V_n$ be the item values in non-increasing order. Since \(z\), together with the
items of \(DG(1)\), forms a subset of the items of the instance, the sorted value sequence of this subset is coordinate-wise dominated by \((V_1,V_2,\ldots,V_n)\).  Let \(c_1\ge c_2\ge \cdots \ge c_n\) be the required
coefficient vector for this subset, padded with zeros. That is, the nonzero coefficients are \(\rho\) for \(z\) and \(\delta_i\) for each \(i\in DG(1)\). By the bounds above and by the assumptions of the lemma,
\[
    \sum_{j=1}^\ell c_j
    \le
    \sum_{j=1}^\ell P_j(a,b)
    \qquad\text{for every }\ell\ge 1.
\]
Indeed, the first prefix is bounded by \(\rho\), the second by \(\rho+\delta_{1/2}\), and every larger prefix by \(\rho+2\delta_{1/2}\), while the assumptions give the corresponding lower bounds for the first three prefixes of the \(P_j\)'s and all later \(P_j\)'s
are nonnegative.

By the summation-by-parts form of majorization, since
\(V_1\ge V_2\ge \cdots \ge V_n\ge 0\), this prefix domination implies
\[
    \sum_{j=1}^n P_j(a,b)V_j
    \ge
    \sum_{j=1}^n c_jV_j .
\]
By rearrangement, the right-hand side is at least the value of assigning the coefficient \(\rho\) to \(z\) and coefficient \(\delta_i\) to each \(i\in DG(1)\). Therefore
\[
    \sum_{j=1}^n P_j(a,b)V_j
    \ge
    \rho v_z+\sum_{i\in DG(1)}\delta_i v_i .
\]

Now let \(p_j\) denote the actual probability that Phase II selects the \(j\)-th largest-value item. By Lemma~\ref{lem:phase-two-rank-probabilities},
\[
    p_j\ge P_j(a,b)
    \qquad\text{for every }j\ge1.
\]
Since Phase II selects at most one item, its expected contribution is
\[
    \E{}{\textnormal{Phase II}}
    =
    \sum_{j=1}^n p_jV_j
    \ge
    \sum_{j=1}^n P_j(a,b)V_j .
\]
Combining this with the previous display gives
\[
    \E{}{\textnormal{Phase II}}
    \ge
    \rho v_z+\sum_{i\in DG(1)}\delta_i v_i .
\]

Combining the bounds for Phase II and Phase III, we obtain
\[
\begin{aligned}
    \E{}{\textnormal{Phase II}+\textnormal{Phase III}}
    &\ge
    \rho v_z+\sum_{i\in DG(1)}\delta_i v_i
    +
    \rho v(DG(1))-\sum_{i\in DG(1)}\delta_i v_i  \\
    &=
    \rho\left(v(DG(1))+v_z\right).
\end{aligned}
\]